%% file: main-JLC.tex
\documentclass[a4paper, fleqn]{article}

\usepackage{authblk}
\usepackage{amsmath,amssymb}
\usepackage{amsthm}
\usepackage{proof}
\usepackage{stmaryrd}
\usepackage{latexsym}
\usepackage{times}
\usepackage{url}
\usepackage{color}
\usepackage{xspace}
\usepackage{tikz}
\usepackage{relsize}
\RequirePackage{txfonts}
\usepackage{microtype}
\usepackage[show]{ed}
\usepackage{lineno}
\usepackage{hyperref}
\usepackage{graphicx}
\usepackage{relsize}
\usepackage{bussproofs}

\input{mycommands.sty}

\newtheorem{theorem}{Theorem}[section]
\newtheorem{lemma}[theorem]{Lemma}
\newtheorem{proposition}[theorem]{Proposition}
\newtheorem{corollary}[theorem]{Corollary}
\newtheorem{example}[theorem]{Example}
\newtheorem{definition}[theorem]{Definition}

\title{Separability and harmony in ecumenical systems\thanks {Pereira, Pimentel and Sales have been partially supported by CAPES and CNPq.}}
\author[1]{Sonia Marin}
\author[2]{Luiz Carlos Pereira}
\author[3]{Elaine Pimentel}
\author[4]{Emerson Sales}

\affil[1]{Department of Computer Science, Birmingham University, UK\\
s.marin@ucl.ac.uk}
\affil[2]{Department of Philosophy, PUC-Rio \& UERJ\\
luiz@inf.puc-rio.br}
\affil[3]{Department Computer Science, UCL, UK\\
elaine.pimentel@gmail.com}
\affil[4]{Gran Sasso Science Institute, Italy\\
emerson.sales@gssi.it}

\date{}
\begin{document}
\maketitle

\begin{abstract}
The quest of smoothly combining logics so that connectives from classical and intuitionistic logics can co-exist in peace has been a fascinating topic of research for decades now. In 2015, Dag Prawitz proposed a  natural deduction system for an ecumenical first-order logic. We start this work by proposing a {\em pure} sequent calculus version for it, in the sense that connectives are introduced without the use of other connectives. For doing this, we extend sequents with an extra context, the stoup, and define the ecumenical notion of polarities. Finally, we smoothly extend these ideas for handling modalities, presenting pure labeled and nested systems for ecumenical modal logics. 


\bigskip
\noindent
\textbf{Keywords} \quad Ecumenical systems; modalities; nested systems; labeled systems; cut-elimination; polarities.

\end{abstract}

\section{Introduction}\label{sec:intro}
\input{introduction}

\section{The system $\LCE$}\label{sec:LCE}
\input{LCE}

\subsection{Ecumenical consequence and stoup}\label{sec:interp}

\input{interpretation}



\section{Correctness of the systems}\label{sec:corr}
\input{correctness}

\section{Ecumenical modalities}\label{sec:modal}
\input{modal}

\section{A nested system for ecumenical modal logic}\label{sec:nested}
\input{nested}

\subsection{Soundness and completeness}\label{sec:sound-comp}
\input{sound-comp}

\section{Fragments, axioms and extensions}\label{sec:frag}
\input{fragments}

\subsection{About axioms and extensions}\label{sec:ext}
\input{extensions}

\section{Related  and future work} \label{sec:conc}
\input{conclusion}


\bibliographystyle{abbrv}
\bibliography{references}
\end{document}

%% file: introduction.tex

$\LC$~\cite{DBLP:journals/mscs/Girard91} is a sequent system for classical logic that separates the rules for {\em positive} and {\em negative} formulas, being a precursor of the notion of {\em focusing} in sequent systems~\cite{andreoli01apal}.
The idea is that right rules for positive formulas are applied in the {\em stoup}, which is a differentiated context, where formulas are focused on. Negative formulas, on the other hand, are stored in a {\em classical context}, where they can be eagerly decomposed.

Sequents with one stoup have the form 
$\Seq{\Gamma}{\Pi}{\Delta}$,\footnote{It should be mentioned that, in $\LC$, sequents have the one-sided presentation -- the left context is not present. Also, in systems like Girard's $\LU$~\cite{DBLP:journals/apal/Girard93}, sequents have two stoups and linear contexts. Here we adopted the simpler possible version for sequents with stoup supporting the intuitionistic setting and avoiding structural rules in classical contexts.} 
where $\Gamma,\Delta$ are sets and $\Pi$, the stoup, is a multiset containing at most one formula. In $\LC$, the meaning of these contexts is the following.
\begin{itemize}
\item[-] $\Gamma$ is the usual classical left context in well known sequent systems for classical and intuicionistic first order logics, like $\LK$ and $\LJ$~\cite{troelstra96bpt}.
\item[-] The stoup $\Pi$ is a differentiated context, where  positive formulas are ``worked on''. 
In a bottom-up read, a positive formula can be {\em chosen} from the classical right context to populate the stoup using the dereliction rule
$$
\infer[\D]{\Seq{\Gamma}{ \cdot}{\Delta,P}}{\Seq{\Gamma}{P}{\Delta,P}} 
$$
\item[-] The (classical) right context $\Delta$ carries the information of subformulas
of negative formulas $N$, in the sense that
$N=\neg N'$.
This means that $N\in\Delta$ can be interpreted as $N'\in\Gamma$. Negative formulas are added to the classical context via the store rule
$$
\infer[\store]{\Seq{\Gamma}{N}{\Delta}}{\Seq{\Gamma}{\cdot}{\Delta,N}} 
$$
\end{itemize}

It is interesting to note that, while the sequent $\Seq{\Gamma}{\Pi}{\Delta}$ is intuitionistically interpreted as 
$
\Gamma,\neg\Delta\seq\Pi
$, 
it only has a classical interpretation in $\LC$ if the stoup $\Pi$ is {\em empty}.
Moreover, the stoup in $\LC$ is {\em persistent}, in the sense that, after applying dereliction $\D$ over a positive formula $P$ in the bottom-up reading, the stoup is emptied  only when either $P$  is  totally consumed, or a negative subformula is reached -- in which case it is stored in the classical right context. 

The ecumenical systems we will study in this paper have a quite different behavior, since they are {\em intuitionistic} in nature. Hence the use of stoups will mix some of the characteristics of $\LC$ with intuitionistic systems featuring stoup such as, for example, Herbelin's $\LJT$ and $\LJQ$~\cite{DBLP:conf/csl/Herbelin94,DBLP:journals/logcom/DyckhoffL07}. The base difference is that, in the ecumenical formulation, stoups cannot be persistent since, otherwise, the logic would not be complete. 

Several approaches have been proposed for combining intuitionistic and classical logics (see \eg~\cite{DBLP:conf/frocos/CerroH96,DBLP:journals/apal/LiangM11,DBLP:journals/Dowek16a}), many of them inspired by  Girard's  polarised system $\LU$ (\cite{DBLP:journals/apal/Girard93}).
More recently, 
Prawitz chose a  completely different approach by proposing a natural deduction ecumenical system~\cite{DBLP:journals/Prawitz15}. While it also took into account meaning-theoretical
considerations, it is more focused on investigating the philosophical significance of the fact that classical logic can be translated into intuitionistic logic.

In this paper,  we will proceed with a careful study of Prawitz' ecumenical system under the view of Girard's original idea of stoup, for separating the intuitionistic from the classical behaviors. This will also allow for a 
a first-order ecumenical system that avoids the use of negations in the formulation of rules. Such systems are called {\em pure} or {\em separable}~\cite{Murzi2018}, in the sense that connectives are introduced without the use of other connectives, hence giving a clearer notion of the {\em meaning} for that connective. This goes straight into the direction first pointed by Prawitz, and adopted by the Proof-theoretic semantics' school~\cite{DBLP:journals/synthese/KahleS06}. Finally, we will extend this notion to modalities.

In this work, we bring new basis for ecumenical systems, where systems and results presented in~\cite{DBLP:conf/dali/MarinPPS20,DBLP:conf/wollic/MarinPPS21} fit smoothly. More specifically, this work improves the {\em op.cit.} in the following ways:
\begin{enumerate}
\item Instead of building the modal system over the sequent presentation~\cite{DBLP:journals/synthese/PimentelPP21} of Prawitz' ecumenical system~\cite{DBLP:journals/Prawitz15}, we propose a new pure first-order ecumenical system. This not only allows for a better proof theoretic view of Prawitz' original proposal, but it also serves as a solid ground for smoothly accommodating modalities.
\item A new pure labeled system for modalities comes naturally in this approach, and the nested system in~\cite{DBLP:conf/wollic/MarinPPS21} is easily proven correct and complete w.r.t. it. 
\item The proof of completeness of the nested system is new, and it does not refer to the axiomatic system.
\end{enumerate}
Under this new perspective, we can start new lively discussions about the nature of formulas and systems.

The rest of the paper is organized as follows: Section~\ref{sec:LCE} introduces the notion of ecumenical systems with stoup (system $\LCE$), and in in Section~\ref{sec:corr} we prove it complete and correct w.r.t. Prawitz' ecumenical system ($\LE$). This involves a non-trivial use of polarities, as well as a non-standard proof of cut-elimination. We show that, is one is not careful, the quest for {\em purity} ends up in {\em collapsing}; Section~\ref{sec:modal}  extends the propositional fragment of $\LCE$ with modalities, resulting in a new pure labeled ecumenical modal system ($\labEK$);  Section~\ref{sec:nested} brings the nested ecumenical system $\nEK$ from~\cite{DBLP:conf/wollic/MarinPPS21}, which is naturally seen as the label-free counterpart of $\labEK$; Section~\ref{sec:frag} briefly discusses fragments, axioms and extensions and Section~\ref{sec:conc} discusses related and future work, and concludes the paper.

%% file: LCE.tex

In~\cite{DBLP:journals/Prawitz15} Dag Prawitz proposed a natural deduction system where classical and intuitionistic logic could coexist in peace. In this system, 
the classical logician and the intuitionistic logician would share the universal quantifier, conjunction, negation and the constant for the absurd, but they would each have their own existential quantifier, disjunction and implication, with different meanings. Prawitz' main idea is that these different meanings are given by a semantical framework that can be accepted by both parties.  

The sequent system $\LE$ (depicted in Fig.~\ref{Fig:LE}) was presented in~\cite{DBLP:journals/synthese/PimentelPP21} as the sequent counterpart of Prawitz' natural deduction system. 

The language $\Lscr$ used for ecumenical systems is described as follows. We will use a subscript $c$ for the classical meaning and $i$ for the intuitionistic, dropping such subscripts when formulas/connectives can have either meaning. 

Classical and intuitionistic  n-ary predicate symbols ($p_{c}, p_{i},\ldots$) co-exist in $\Lscr$ but have different meanings. 
The neutral logical connectives $\{\bot,\neg,\wedge,\forall\}$ are common for classical and intuitionistic fragments, while $\{\iimp,\ivee,\iexists\}$ and $\{\cimp,\cvee,\cexists\}$ are restricted to intuitionistic and classical interpretations, respectively.

$\LE$ has very interesting proof theoretical properties including cut-elimination, together with a Kripke semantical interpretation, that allowed the proposal of a variety of ecumenical proof systems, such as multi-conclusion and nested sequent systems, as well as several  fragments of such systems~\cite{DBLP:journals/synthese/PimentelPP21}.

\begin{figure}[t]
{\sc Intuitionistic and neutral Rules}
$$
\infer[{\wedge L}]{A \wedge B,\Gamma \seq 
C}{
A, B,\Gamma \seq C} 
\quad
\infer[{\wedge R}]{\Gamma \seq
A \wedge B}{\Gamma
\seq A \quad \Gamma \seq  B} 
\quad
\infer[{\vee_i L}]{A \vee_i B,\Gamma \seq
C}{A,\Gamma
\seq C \quad B,\Gamma \seq C} 
$$
$$
\infer[{\vee_i R_j}]{\Gamma \seq  A_1\vee_i A_2}{\Gamma
\seq A_j} 
\quad
\infer[{\iimp L}]{\Gamma, A\iimp B \seq
C}
{A\iimp B,\Gamma \seq A \quad B, \Gamma  \seq 
C} 
\quad
\infer[{\iimp R}]{\Gamma \seq 
A\iimp B}{\Gamma, A \seq  B}
$$
$$
\infer[{\neg L}]{\neg A,\Gamma \seq \bot
}{ \neg A,\Gamma \seq  A}
\quad
\infer[{\neg R}]{\Gamma \seq 
\neg A}{\Gamma, A \seq  \bot}
\quad
\infer[{\bot L}]{\bot,\Gamma \seq A}{}
$$
$$
\infer[\forall L]{\forall x.A,\Gamma \seq C }{A[y/x],\forall x.A,\Gamma \seq C}
\quad
\infer[\forall R]{\Gamma \seq \forall x.A }{\Gamma \seq A[y/x]}
\quad
\infer[\exists_i L]{ \exists_ix.A,\Gamma \seq C}{A[y/x],\Gamma \seq C}
\quad
\infer[\exists_i R]{\Gamma \seq \exists_ix.A }{\Gamma \seq A[y/x]}
$$
{\sc Classical rules}

\vspace{.3cm}
\resizebox{\textwidth}{!}{$
\infer[\vee_cL]{ A \vee_c B,\Gamma \seq\bot }{A,\Gamma \seq \bot & B, \Gamma
\seq \bot} 
\quad
\infer[\vee_cR]{\Gamma
\seq \ A \vee_c B}{\Gamma, \neg A , \neg B \seq \bot}
\quad
\infer[\cimp L]{ A \rightarrow_c B,\Gamma \seq\bot }{ A \rightarrow_c B,\Gamma \seq A & B,\Gamma \ \seq
\bot}
$}
$$
\infer[\cimp R]{\Gamma \seq A
\rightarrow_c B}{\Gamma, A , \neg B \seq  \bot} 
\quad
\infer[L_c]{p_c,\Gamma\seq\bot}{p_i,\Gamma\seq \bot } 
\quad
\infer[R_c]{\Gamma \seq p_c}{\Gamma,\neg p_i \seq\bot} 
$$
$$
\infer[\exists_c L]{\exists_cx.A,\Gamma \seq \bot}{A[y/x],\Gamma \seq  \bot}
\quad
\infer[\exists_c R]{\Gamma \seq \exists_cx.A}{\Gamma, \forall x.\neg A \seq \bot}
$$
{\sc Initial, cut and Structural Rules}
$$
\infer[\init]{p_i,\Gamma \seq p_i}{} 
\qquad 
\infer[\cut]{\Gamma \seq C}{\Gamma \seq A & A,\Gamma \seq C} 
\qquad 
\infer[{\W}]{\Gamma \seq A}{\Gamma \seq \bot}
$$
\caption{Ecumenical sequent system  $\LE$. In rules 
$\forall R,\iexists L,\cexists L$, the eigenvariable $y$ is fresh; $p$ is atomic.}\label{Fig:LE}
\end{figure}

%% file: interpretation.tex
Denoting by $\vdash_{\mathsf{S}} A$ the fact that the formula $A$ is a theorem in the proof system $\mathsf{S}$, the following  theorems are easily provable in $\LE$:
\begin{enumerate}
\item\label{and} $\vdash_{\LE}  (A \vee_c B) \leftrightarrow_i  \neg (\neg A \wedge \neg B)$;
\item\label{or} $\vdash_{\LE} (A \to_c B) \leftrightarrow_i \neg (A \wedge \neg B)$; 
\item\label{foe1} $\vdash_{\LE} (\exists_cx.A) \leftrightarrow_i \neg (\forall x.\neg A)$.
\end{enumerate}
There equivalences are of interest since they relate the classical
and the neutral operators: the classical connectives can be defined using negation, conjunction, and the universal quantifier. 
On the other hand,
\begin{enumerate}
\setcounter{enumi}{3}
\item $\vdash_{\LE} (\neg\neg A) \to_c A $ but  $\not\vdash_{\LE} (\neg\neg A) \to_i A $ in general;
\item\label{mp} $\vdash_{\LE} (A\wedge(A  \to_i B)) \to_i B$ but $\not\vdash_{\LE} (A\wedge(A  \to_c B)) \to_i B $ in general;
\item\label{foe2} $\vdash_{\LE} \forall x.A \iimp  \neg \exists_cx.\neg A $ but $\not\vdash_{\LE} \neg\exists_cx.\neg A \to_i \forall x.A$ in general.
\end{enumerate}
\vspace{-.2cm}
Observe that~(\ref{foe1}) and~(\ref{foe2}) reveal the asymmetry between definability of quantifiers: while the classical existential can be defined from the universal quantification, the other way around is not true, in general. This is closely related with the fact that, proving $\forall x.A$ from $\neg\exists_cx.\neg A$ depends on $A$ being a classical formula. We will come back to this in Section~\ref{sec:modal}.

On its turn,  the following result states that logical consequence in $\LE$ is intrinsically intuitionistic. 
\begin{proposition}[\cite{DBLP:journals/synthese/PimentelPP21}]\label{prop:ev}
$\Gamma\vdash B$ is provable in $\LE$ iff $\vdash_{\LE}\bigwedge\Gamma\iimp B$.
\end{proposition}

To preserve  the ``classical behavior'', \ie, to satisfy all the principles of classical logic \eg\ {\em modus ponens} and the {\em classical reductio}, it is sufficient that the main operator of the formula
will be eventually classical~\cite{DBLP:journals/synthese/PimentelPP21}. Thus,  ``hybrid" formulas, \ie, formulas that contain
 classical and intuitionistic operators may have a classical behavior. Formally,
\begin{definition}\label{def:eec}
{\em Eventually externally classical} (\eec\ for short) formulas  are given by the following grammar
$$
\begin{array}{lcl}
A^{ec} & := & A^{c}  \mid A^{ec}\wedge A^{ec}  \mid \forall x. A^{ec} \mid A\iimp A^{ec} \mid \neg A
\end{array}
$$
where $A$ is any formula and $A^{c}$ is an {\em externally classical} formula given by 
$$
\begin{array}{lcl}
A^{c} & := & p_{c}  \mid \bot \mid A\cvee A   \mid A\cimp A \mid \cexists x. A
\end{array}
$$
\end{definition}
For \eec\ formulas we can prove the following theorems
\begin{enumerate}
\setcounter{enumi}{6}
\item $\vdash_{\LE} (A\wedge(A  \to_c B^{ec}))  \iimp B^{ec} $.
\item $\vdash_{\LE} \neg\neg B^{ec}\iimp B^{ec}$.\label{dn}
\item $\vdash_{\LE} \neg\exists_cx.\neg B^{ec} \to_i \forall x.B^{ec}$.\label{ue}
\end{enumerate}
More generally, notice that all classical right rules as well as the right rules for the neutral connectives in $\LE$ are invertible. Since invertible rules can be applied eagerly when proving a sequent, this entails that \eec\ formulas can be eagerly decomposed.
As a consequence, the ecumenical entailment, when restricted to \eec\ succedents (antecedents having an unrestricted form), is classical.
\begin{theorem}[\cite{DBLP:journals/synthese/PimentelPP21}, extended]\label{thm:classical-right}
Let $A^{ec}$ be an eventually externally classical formula and $\Gamma$ be a multiset of ecumenical formulas. Then
\vspace{-.2cm}
$$\vdash_{\LE}\bigwedge\Gamma\cimp A^{ec} \mbox{ iff } \vdash_{\LE}\bigwedge\Gamma\iimp A^{ec}.$$ 
\end{theorem}
This sums up well, proof theoretically, the {\em ecumenism} of Prawitz' original proposal: consequence relations are intrinsically intuitionistic, but have a classical behavior when proving a formula that eventually will behave classically.

Moreover, observe that, from a proof $\pi$ of $\Gamma\seq A$ in $\LE$, we can derive $\Gamma,\neg A\seq \bot$:
$$
\infer[\neg L]{\Gamma,\neg A\seq \bot}{\deduce{\Gamma,\neg A \seq A}{\pi^w}}
$$
where $\pi^w$ is the weakened version of $\pi$.
The other direction does not hold since
$\vdash_{\LE} \neg \neg A\not\seq A$, in general.
However, for \eec\ formulas the converse  also holds.
\begin{proposition}
If $\Gamma,\neg A^{ec}\seq \bot$ is provable in $\LE$ so it is $\Gamma\seq A^{ec}$. 
\end{proposition}
\begin{proof}
Since $\vdash_{\LE} \neg \neg A^{ec}\seq A^{ec}$ (see~\ref{dn}), then
$$
\infer[\cut]{\Gamma\seq A^{ec}}
{\infer{\Gamma\seq \neg\neg A^{ec}}{\Gamma,\neg A^{ec}\seq \bot}&\deduce{\neg\neg A^{ec},\Gamma\seq A^{ec}}{}}
$$
\end{proof}
This corroborates the idea that, in an ecumenical system with stoup, formulas in the classical context should hold classical subformulas of \eec\ formulas.
The stoup, on the other hand, would carry the {\em intuitionistic} or {\em neutral} information. 

We are now ready to describe the ecumenical system with stoup, where the connections between the ``primitive" sequent calculus in $\LE$ and the ``pure" sequent calculus in $\LCE$ is established as follows:
\begin{itemize}
\item[-] A sequent of the form $\Gamma , \neg \Delta   \seq \Pi$ will be translated as $\Gamma  \seq \Delta ; \Pi$ for some set $\Delta$ of negated formulas.
\item[-] A sequent of the form $\Gamma  \seq \Delta ; \Pi$ will be translated as $\Gamma ,  \neg \Delta  \seq \Pi$.
\item[-] The empty stoup will be translated as $\bot$.
\end{itemize}
As already mentioned, formulas will move over contexts depending on the {\em polarity}.
\begin{definition}\label{def:pol}
		An ecumenical formula is called {\em negative} if its main connective is classical or the negation, and {\em positive} otherwise (we will use $N$ for negative and $P$ for positive formulas).
\end{definition}
Figure~\ref{fig:LCE} brings the rules for the ecumenical pure systems with 
stoup ($\LCE$).
Observe that rules in $\LCE$ determine positive/negative phases in derivations, and the dynamic of the rules for classical connectives in $\LCE$ is as follows: Negative formulas in the classical contexts are eagerly decomposed; if a positive formula in the right context  is chosen to be worked on, it is placed in the stoup $\Pi$, and treated intuitionistically.

\begin{figure}[t]
{\sc Intuitionistic and neutral Rules}
$$
\infer[{\wedge L}]{\Seq{\Gamma, A \wedge B}{ \Pi} {\Delta}}{\Seq{\Gamma,
A, B}{ \Pi}{\Delta}} 
\quad
\infer[{\wedge R}]{\Seq{\Gamma}
{A \wedge B}{\Delta}}{\deduce{\Seq{\Gamma} {A}{\Delta}}{} &\deduce{\Seq{\Gamma}{  B}{\Delta}}{}} 
\quad
\infer[{\vee_i L}]{\Seq{\Gamma, A \vee_i B}{ \Pi}{\Delta}
}{\deduce{\Seq{\Gamma, A}{ \Pi}{\Delta}}{} &\deduce{\Seq{\Gamma,B}{\Pi}{\Delta}}{}} 
$$
$$
\infer[{\vee_i R_j}]{\Seq{\Gamma}{ A_1\vee_i A_2}{\Delta}}{\Seq{\Gamma}{ A_j}{\Delta}} 
\quad
\infer[{\iimp L}]{\Seq{\Gamma, A\iimp B}{
\Pi}{\Delta}}
{\Seq{\Gamma, A\iimp B}{ A}{\Delta} & \Seq{\Gamma , B} 
{\Pi}\Delta} 
\quad
\infer[{\iimp R}]{\Seq{\Gamma}
{A\iimp B}{\Delta}}{\Seq{\Gamma,A}{ B}{\Delta}}
$$
$$
\infer[{\neg L}]{\Seq{\Gamma,\neg A}{\cdot}{\Delta}
}{\Seq{\Gamma, \neg A}{  A}{\Delta}}
\qquad
\infer[{\neg R}]{\Seq{\Gamma} 
{\cdot}{\Delta,\neg A}}{\Seq{\Gamma,A}{  \cdot}{\Delta}}
\qquad
\infer[\exists_i L]{\Seq{\Gamma, \exists_ix.A}{ \Pi}{\Delta}}{\Seq{\Gamma, A[y/x]}{ \Pi}{\Delta}}
$$
$$
\infer[\exists_i R]{\Seq{\Gamma}{ \exists_ix.A}{\Delta }}{\Seq{\Gamma}{ A[y/x]}{\Delta}}
\quad
\infer[\forall L]{\Seq{\Gamma,\forall x.A}{ \Pi}{\Delta}}{\Seq{\Gamma,\forall x.A,A[y/x]}{ \Pi}{\Delta}}
\quad
\infer[\forall R]{\Seq{\Gamma}{ \forall x.A}{\Delta}}{\Seq{\Gamma}{ A[y/x]}{\Delta}}
$$
{\sc Classical Rules}

\vspace{.3cm}
\resizebox{\textwidth}{!}{$
\infer[\vee_cL]{\Seq{\Gamma,  A \vee_c B}{\cdot}{\Delta}}{\Seq{\Gamma,A}{ \cdot}{\Delta} & \Seq{\Gamma,B}{ \cdot}{\Delta}} 
\quad
\infer[\vee_cR]{\Seq{\Gamma}{ \cdot}{ A \vee_c B,\Delta}}{\Seq{\Gamma}{ \cdot}{A,B,\Delta}}
\quad
\infer[\cimp L]{\Seq{\Gamma, A \rightarrow_c B}{\cdot}{\Delta} }{\Seq{\Gamma, A \rightarrow_c B}{ A}{\Delta}&\Seq{\Gamma, B}{\cdot}{\Delta}}
$}
$$
\infer[\cimp R]{\Seq{\Gamma} {\cdot}{A
\rightarrow_c B, \Delta}}{\Seq{\Gamma, A}{  \cdot}{B,\Delta}} 
\qquad
\infer[L_c]{\Seq{\Gamma, p_c}{\cdot}{\Delta}}{\Seq{\Gamma,p_i}{ \cdot}{\Delta}}
\qquad
\infer[R_c]{\Seq{\Gamma}{ \cdot}{p_c,\Delta}}{\Seq{\Gamma}{ \cdot}{p_i,\Delta}}
$$
$$
\infer[\exists_c L]{\Seq{\Gamma ,\exists_cx.A}{ \cdot}{\Delta}}{\Seq{\Gamma,  A[y/x]}{ \cdot}{\Delta}}
\qquad
\infer[\exists_c R]{\Seq{\Gamma}{ \cdot}{\exists_cx.A,\Delta}}{\Seq{\Gamma}{\cdot}{A[y/x],\exists_cx.A,\Delta}}
$$
{\sc Initial, Decision and Structural Rules}
$$
\begin{array}{c@{\qquad}c@{\qquad}c@{\qquad}c}
\infer[\init]{\Seq{\Gamma,p_i}{ p_i}{\Delta}}{} 
&
\infer[\D]{\Seq{\Gamma}{ \cdot}{P,\Delta}}{\Seq{\Gamma}{ P}{P,\Delta}} 
&
\infer[\store]{\Seq{\Gamma}{ N}{\Delta}}{\Seq{\Gamma}{ \cdot}{N,\Delta}} 
&
\infer[{\W}]{\Seq{\Gamma}{A}{\Delta}}{\Seq{\Gamma}{\cdot}{\Delta}}
\end{array}
$$
{\sc Cut Rules}
$$
\begin{array}{c@{\qquad}c}
\infer[P-\cut]{\Seq{\Gamma}{ \Pi}{\Delta}}{\Seq{\Gamma}{ P}{\Delta} & \Seq{P,\Gamma}{ \Pi}{\Delta}} 
&
\infer[N-\cut]{\Seq{\Gamma}{ \Pi}{\Delta}}{\Seq{\Gamma}{\Pi^*}{\Delta, N} & \Seq{N,\Gamma}{ \Pi}{\Delta}} 
\end{array}
$$
\caption{Ecumenical pure system $\LCE$. In rules 
$\forall R$, $\exists_c L$ and $\exists_i L$, $y$ is fresh; $N$ is negative and $P$ is positive; $p$ is atomic; $\Pi^*$ is either empty or some $P\in\Delta$.}\label{fig:LCE}
\end{figure}

The analogy with focusing~\cite{andreoli01apal} stops there, though: explicit weakening in the stoup is needed for completeness, as the next example shows.
\begin{example}\label{ex:weak}
The sequent 
$$\neg B,A\cimp B, A\vdash C$$
is provable in $\LE$, where the succedent $C$ is weakened. This means that, in $\LCE$,  the following sequents should be provable
$$
\Seq{\neg B,A\cimp B, A}{C}{\cdot}\qquad \Seq{\cdot}{C}{B,A\wedge\neg B,\neg A}
$$
But the stoup is necessarily erased in the process.

Observe that, in $\LC$, sequents with non-empty stoup do not have a classical interpretation.  In fact, {\bf none} of the sequents above are provable in $\LC$, if $C$ is a positive formula.
\end{example}

%% file: correctness.tex

We start by stating standard proof theoretic results.
\begin{lemma}\label{lemma:wc}
In $\LCE$:
\begin{enumerate}
\item[i.] The rules $\cvee L, \cvee R,\cimp L, \cimp R,\neg L,\neg R, L_c, R_c,\cexists L,\cexists R$ and $\D$ are {\em invertible}, that is, in any application of such rules, if the conclusion is a provable nested sequent so are the premises.
\item[ii.] The rules $\wedge L, \wedge R, \ivee L, \iimp R,\iexists L, \forall L,\forall R$ and $\store$ are {\em totally invertible},  that is, they are invertible and can be applied in any contexts. 
\item[iii.] Classical weakening and contraction are admissible 
		$$
		\infer[\W_c]{\Seq{\Gamma,\Gamma'}{\Pi}{\Delta,\Delta'}}{\Seq{\Gamma}{\Pi}{\Delta}} \qquad 
		\infer[\C_c]{\Seq{\Gamma,\Gamma'}{\Pi}{\Delta,\Delta'}}{\Seq{\Gamma,\Gamma',\Gamma'}{\Pi}{\Delta,\Delta',\Delta'}}
		$$
\item[iv.] The general form of initial axioms are admissible \label{init}
$$
\infer[\ginit]{\Seq{\Gamma,A}{A}\Delta}{} \qquad 
\infer[\gcinit]{\Seq{\Gamma,A}{\Pi}{\Delta,A}}{} 
$$
	\end{enumerate}
\end{lemma}
\begin{proof} The proofs are by standard induction on the height of derivations. The proof of admissibility of $\W_c$ does not depend on any other result, while the admissibility of $\C_c$ depends on the invertibility results and the admissibility of weakening. 

The proof of admissibility of the general initial axioms is by mutual induction. 
Below we show the cases for quantifiers 
where, by induction hypothesis, instances of the axioms hold for the premises.
\[
\infer[\cexists L]{\Seq{\cexists x.A,\Gamma}\cdot{\cexists x.A,\Delta}}
{\infer[\cexists R]{\Seq{A[y/x],\Gamma}\cdot{\cexists x.A,\Delta}}
{\infer[\gcinit]{\Seq{A[y/x],\Gamma}\cdot{\cexists x.A,A[x/y],\Delta}}{}}}
\qquad
\infer[\forall R]{\Seq{\forall x.A,\Gamma}{\forall x.A}{\Delta}}
{\infer[\forall L]{\Seq{\forall x.A,\Gamma}{A[y/x]}{\Delta}}
{\infer[\ginit]{\Seq{\forall x.A,A[y/x],\Gamma}{A[y/x]}{\Delta}}{}}}  
\]
\end{proof}

The following shows that $\LCE$ is correct and complete w.r.t. $\LE$.
\begin{theorem} The sequent 
$\Seq{\Gamma}{\Pi}{\Delta}$ is provable in $\LCE$ iff $\Gamma,\neg\Delta\seq\Pi$ is provable in $\LE$. 
\end{theorem}
\begin{proof}
The only interesting cases are the ones involving classical connectives. 
\begin{itemize}
\item[-] Case $\cvee R$. Suppose that $\Seq{\Gamma}{\cdot}{\Delta,A\cvee B}$ is provable in 
$\LCE$ with proof 
$$
\infer[\cvee R]{\Seq{\Gamma}{\cdot}{A \vee_c B,\Delta}}{\Seq{\Gamma}{\cdot}{A,B,\Delta}}
$$
By inductive hypothesis, $\Gamma,\neg A,\neg B,\neg\Delta \seq \bot$ is provable in $\LE$. Hence
$$
\infer[\neg L]{\Gamma,\neg (A \vee_c B),\neg\Delta \seq \bot}
{\infer[\cvee R]{\Gamma,\neg\Delta \seq A \vee_c B}
{\Gamma,\neg A,\neg B,\neg\Delta \seq \bot}}
$$
On the other hand, suppose that $\Gamma\seq A\cvee B$ is provable in $\LE$ with proof
$$
\infer[\cvee R]{\Gamma
\seq A \vee_c B}{\Gamma,\neg A,\neg B \seq \bot}
$$
By inductive hypothesis, $\Seq{\Gamma}{ \cdot}{ A, B}$ is provable in $\LCE$. Thus
$$
\infer[\store]{\Seq{\Gamma}{A\cvee B}{\cdot}}
{\infer[\cvee R]{\Seq{\Gamma}{ \cdot}{A\cvee B}}
{\deduce{\Seq{\Gamma}{ \cdot}{A,B}}{}}}
$$
\item[-] Case $\cimp R$. Suppose that $\Seq{\Gamma}{\cdot}{\Delta,A\cimp B}$ is provable in 
$\LCE$ with proof 
$$
\infer[\cimp R]{\Seq{\Gamma}{\cdot}{A \cimp B,\Delta}}{\Seq{\Gamma,A}{\cdot}{B,\Delta}}
$$
By inductive hypothesis, $\Gamma,A,\neg B,\neg\Delta \seq \bot$ is provable in $\LE$. Hence
$$
\infer[\neg L]{\Gamma,\neg (A \cimp B),\neg\Delta \seq \bot}{
\infer[\cimp R]{\Gamma,\neg\Delta \seq A \cimp B}
{\Gamma,A,\neg B,\neg\Delta \seq \bot}}
$$
On the other hand, suppose that $\Gamma\seq A\cimp B$ is provable in 
$\LE$ with proof 
$$
\infer[\cimp R]{\Gamma\seq A\cimp B}{\Gamma, A,\neg B\seq\bot}
$$
By inductive hypothesis, $\Seq{\Gamma,A}{\cdot}{B}$ is provable in $\LCE$. Hence
$$
\infer[\store]{\Seq{\Gamma}{A\cimp B}\cdot}
{\infer[\cimp R]{\Seq{\Gamma}\cdot{A\cimp B}}
{\deduce{\Seq{\Gamma,A}\cdot{B}}{}}}
$$

\item[-] Case $\cexists R$. Suppose that $\Seq{\Gamma}{\cdot}{\cexists x. A,\Delta}$ is provable in 
$\LCE$ with proof 
$$
\infer[\cimp R]{\Seq{\Gamma}{\cdot}{\cexists x. A,\Delta}}{\Seq{\Gamma}{\cdot}{A[t/x],\cexists x. A,\Delta}}
$$
By inductive hypothesis, $\Gamma,\neg A[y/x],\neg\cexists x. A,\neg\Delta \seq \bot$ is provable in $\LE$. Hence
$$
\infer[\neg L]{\Gamma,\neg \cexists x. A,\neg\Delta \seq \bot}{
\infer[\cimp R]{\Gamma,\neg \cexists x. A,\neg\Delta \seq \cexists x. A}
{
\infer[\forall L]{\Gamma,\neg \cexists x. A, \forall x. \neg A,\neg\Delta \seq \cdot}
{\Gamma,\neg \cexists x. A,\neg A[y/x], \forall x. \neg A,\neg\Delta \seq \cdot}}}
$$
On the other hand, suppose that $\Gamma\seq \cexists x.A$ is provable in 
$\LE$ with proof 
$$
\infer[\cexists R]{\Gamma\seq \cexists x.A}{\deduce{\Gamma, \forall x.\neg A\seq\bot}{\pi}}
$$
By inductive hypothesis, $\Seq{\Gamma,\forall x.\neg A}{\cdot}{\cdot}$ is provable in $\LCE$. Hence
$$
\infer[\store]{\Seq{\Gamma}{\cexists x.A}\cdot}
{\infer[N-\cut]{\Seq{\Gamma}\cdot{\cexists x.A}}
{\infer[\neg R]{\Seq{\Gamma}\cdot{\cexists x.A,\neg\forall x.\neg A}}
{\deduce{\Seq{\Gamma,\forall x.\neg A}\cdot{\cexists x.A}}{\pi_w}}&
\infer={\Seq{\Gamma,\neg\forall x.\neg A}\cdot{\cexists x.A}}{}}}
$$
where $\pi_w$ represents the translation of the weakened version of $\pi$ and the double bars in  the right branch indicates an adapted proof of~\ref{ue}.

\item[-] Case $\cimp L$.
Suppose that $\Seq{\Gamma,A\cimp B}{\cdot}{\Delta}$ is provable in 
$\LCE$ with proof 
$$
\infer[\cimp L]{\Seq{\Gamma,A\cimp B}{ \cdot}{\Delta}}{
\Seq{\Gamma,A\cimp B}{A}{\Delta} & \Seq{\Gamma, B}{\cdot}{\Delta}}
$$
By inductive hypothesis, $\Gamma,A\cimp B,\neg\Delta \seq A$ and $\Gamma, B,\neg\Delta \seq \bot$
are provable in $\LE$. Hence
$$
\infer[\cimp L]{\Gamma,A\cimp B,\neg\Delta \seq \bot}
{\Gamma,A\cimp B,\neg\Delta \seq A &
\Gamma, B,\neg\Delta \seq \bot}
$$
and vice-versa. The other left-rule cases are similar.
\item[-] Case $\D$. Suppose that $\Seq{\Gamma}{\cdot}{\Delta,P}$ is provable in 
$\LCE$ with proof 
$$
\infer[\D]{\Seq{\Gamma}{\cdot}{\Delta,P}}{\Seq{\Gamma}{P}{\Delta,P}}
$$
By inductive hypothesis, $\Gamma,\neg P,\neg\Delta \seq P$ is provable in $\LE$. Hence
$$
\infer[\neg L]{\Gamma,\neg P,\neg\Delta \seq \bot}
{\Gamma,\neg P,\neg\Delta \seq P}
$$
\item[-] Case $\store$. Suppose that $\Seq{\Gamma}{N}{\Delta}$ is provable in 
$\LCE$ with proof 
$$
\infer[\store]{\Seq{\Gamma}{N}{\Delta}}{\Seq{\Gamma}{\cdot}{\Delta,N}}
$$
By inductive hypothesis, $\Gamma,\neg N,\neg\Delta \seq \bot$ is provable in $\LE$. Hence
$$
\infer[\cut]{\Gamma,\neg\Delta \seq N}
{\infer[\neg R]{\Gamma,\neg\Delta \seq \neg\neg N}{\Gamma,\neg\Delta,\neg N \seq \bot}&
\infer={\neg\neg N\seq N}{}}
$$
where the double bar indicates the multiple-steps proof of the fact that, for negative formulas, $N\equiv \neg\neg N$.
\item[-] Cases $\cut$. The $P-\cut$ rule in $\LCE$ trivially corresponds to $\cut$ in $\LE$.
Suppose that $\Seq{\Gamma}{ \Pi}{\Delta}$ is provable in $\LCE$ with proof
$$
\infer[N-\cut]{\Seq{\Gamma}{ \Pi}{\Delta}}{\Seq{\Gamma}{ \Pi^*}{\Delta, N} & \Seq{N,\Gamma}{ \Pi}{\Delta}}
$$
By inductive hypothesis, $\Gamma,\neg N,\neg\Delta \seq \Pi^*$ and $\Gamma, N,\neg\Delta \seq \Pi$ are provable in $\LE$ with proofs $\pi_1$ and $\pi_2$ respectively. 

If $\Pi^*=\cdot$, then
$$
\infer[\cut]{\Gamma,\neg\Delta\seq \Pi}
{\infer[\neg R]{\Gamma,\neg\Delta\seq \neg\neg N}
{\deduce{\Gamma,\neg\Delta,\neg N\seq \bot}{\pi
_1}}&
\infer[\cut]{\neg\neg N,\Gamma,\neg\Delta\seq \Pi}
{\infer={\neg\neg N,\Gamma,\neg\Delta\seq N}{}&
\deduce{\neg\neg N,N,\Gamma,\neg\Delta\seq\Pi}{\pi_2^w}}}
$$
where $\pi_2^w$ is the weakened version of $\pi_2$, and the double bar corresponds to the proof of~\ref{dn}. 

If $\Pi^*=P\in\Delta$, then the left premise derivation is substituted by
$$
\infer[\neg R]{\Gamma,\neg\Delta\seq \neg\neg N}
{\infer[\cut]{\Gamma,\neg\Delta,\neg N\seq \bot}
{\deduce{\Gamma,\neg\Delta,\neg N\seq P}{\pi_1}&
\infer=[P\in\Delta]{P,\Gamma,\neg\Delta,\neg N\seq \bot}{}}}
$$
\end{itemize}
\end{proof}
Observe that, other than the cut rules, $\store$ introduces cut in the translation between proofs from $\LCE$ to $\LE$, while $\cexists R$ introduces cuts in the other way around. Since $\LE$ has the cut-elimination property~\cite{DBLP:journals/synthese/PimentelPP21}, the translation from $\LCE$ to $\LE$ is not problematic. However, for proving cut-completeness from $\LE$ to $\LCE$, we need to prove that the later has the cut-elimination property.

\subsection{Cut-elimination}\label{sub:cut}
A logical connective is called {\em harmonious} in a certain proof system if there exists a certain balance between the rules defining it. For example, in natural deduction based systems, harmony is ensured when introduction/elimination rules do not contain insufficient/excessive amounts of information~\cite{DBLP:conf/ictac/Diaz-CaroD21}. 
In sequent calculus, this property is often guaranteed by the admissibility of a general initial axiom ({\em identity-expansion}) and of the cut rule ({\em cut-elimination})~\cite{DBLP:journals/tcs/MillerP13}. 

In Lemma~\ref{lemma:wc} we proved identity-expansion for $\LCE$.
In the following, we will complete the proof of harmony for $\LCE$, proving that it enjoys the cut-elimination property. This will also guarantee cut-completeness from $\LE$ to $\LCE$, as mentioned above.

Proving admissibility of cut rules in sequent based systems with multiple contexts is often tricky, since the cut formulas can change contexts during cut reductions.
This is the case for $\LCE$. 
The proof is by mutual induction, with inductive measure $(n,m)$ where  $m$ is the cut-height, the cumulative height of derivations above the cut, and $n$ is the ecumenical weight of the cut-formula, defined as

\begin{center}
\begin{tabular}{lc@{\qquad}l}
$\ew(p_i)=\ew(\bot)=0$ & &
$\ew(A\star B)=\ew(A)+\ew(B)+1$ if $\star\in\{\wedge,\iimp,\ivee\}$\\
$\ew(p_c)=4$ & & 
$\ew(\heartsuit A)=\ew(A)+1$ if $\heartsuit\in\{\neg,\iexists x.,\forall x.\}$ \\
$\ew(\cexists x. A)=\ew(A)+4$ & & 
$\ew(A\circ B)=\ew(A)+\ew(B)+4$ if $\circ\in\{\cimp,\cvee\}$
\end{tabular}
\end{center}
\noindent
Intuitively, the ecumenical weight measures the amount of extra information needed (the negations added) to define classical connectives from intuitionistic and neutral ones.
%

\begin{theorem}\label{thm:lcecut}
The rules $N-\cut$ and $P-\cut$ are admissible in $\LCE$.
\end{theorem}
\begin{proof} The dynamic of the proof is the following: cut applications either move up in the proof, \ie\ the cut-height is reduced, or are substituted by simpler cuts of the same kind, \ie\ the ecumenical weight is reduced, as in usual cut-elimination reductions. 
The cut instances alternate between positive and negative (and vice-versa) in the principal cases, where the polarity of the subformulas flip.
We will detail the main cut-reductions.
\begin{itemize}
\item[-] Base cases. Consider the derivation 
$$
\infer[P-\cut]{\Seq{\Gamma}{\Pi}{\Delta}}
{\deduce{\Seq{\Gamma}{p_i}{\Delta}}{\pi}&
\infer[init]{\Seq{\Gamma,p_i}{\Pi}{\Delta}}{}}
$$
If $p_i$ is principal, then $\Pi=p_i$ and the the derivation reduces to $\pi$.

If  $p_i$ is not principal, then there is an  atom $q_i\in\Gamma\cap\Pi$  and  the reduction is a trivial one. Similar analyses hold for $N-\cut$, when the left premise is an instance of $\init$, as well as for the other axioms.
\item[-] Non-principal cases. In all the cases where the cut-formula is not principal in one of the premises, the cut moves upwards. The only exceptions are when: 

- dereliction is applied in the left premise
$$
\infer[N-\cut]{\Seq\Gamma\Pi{\Delta,P}}
{\infer[\D]{\Seq\Gamma\cdot{\Delta,P,N}}
{
\deduce{\Seq\Gamma{P}{\Delta,P,N}}{\pi_1}}
&
\deduce{\Seq{N,\Gamma}\Pi{\Delta,P}}{\pi_2}
} 
$$
In this case, we substitute the version of $N-\cut$ for absorbing the dereliction
$$
\infer[N-\cut]{\Seq\Gamma{\Pi}{\Delta,P}}
{\deduce{\Seq\Gamma{P}{\Delta,P,N}}{\pi_1}
&
\deduce{\Seq{N,\Gamma}\Pi{\Delta,P}}{\pi_2}
}
$$
- weakening is applied in the left premise
$$
\infer[N-\cut]{\Seq\Gamma\Pi{\Delta,P}}
{\infer[\W]{\Seq\Gamma{P}{\Delta,P,N}}
{
\deduce{\Seq\Gamma{\cdot}{\Delta,P,N}}{\pi_1}}
&
\deduce{\Seq{N,\Gamma}\Pi{\Delta,P}}{\pi_2}
} 
$$
In this case, we substitute the version of $N-\cut$ for absorbing the weakening
$$
\infer[N-\cut]{\Seq\Gamma{\Pi}{\Delta,P}}
{\deduce{\Seq\Gamma{\cdot}{\Delta,P,N}}{\pi_1}
&
\deduce{\Seq{N,\Gamma}\Pi{\Delta,P}}{\pi_2}
}
$$

\item[-] Principal cases. If the cut formula is principal in both premises, then  we need to be extra-careful with the polarities. We show two most representative cases.

- $N=P\cimp Q$, with $P,Q$ positive.
$$
\infer[N-\cut_0]{\Seq\Gamma\cdot{\Delta}}
{\infer[\cimp R]{\Seq\Gamma\cdot{\Delta,P\cimp Q}}
{\deduce{\Seq{\Gamma,P}\cdot{\Delta,Q}}{\pi_1}}&
\infer[\cimp L]{\Seq{\Gamma,P\cimp Q}\cdot{\Delta}}
{\deduce{\Seq{\Gamma,P\cimp Q}{P}{\Delta}}{\pi_2}&
\deduce{\Seq{\Gamma,Q}\cdot{\Delta}}{\pi_3}}}
$$
reduces to

\resizebox{\textwidth}{!}{$
\infer[N-\cut_1]{\Seq\Gamma\cdot{\Delta}}
{\infer[\neg R]{\Seq\Gamma\cdot{\Delta,\neg Q}}
{\deduce{\Seq{\Gamma,Q}\cdot{\Delta}}{\pi_3}}&
\infer[P-\cut]{\Seq{\Gamma,\neg Q}\cdot{\Delta}}
{\infer[N-\cut_2]{\Seq{\Gamma,\neg Q}{P}{\Delta}}
{\infer[\cimp R]{\Seq{\Gamma,\neg Q}{\cdot}{\Delta,P\cimp Q}}
{\deduce{\Seq{\Gamma,\neg Q, P}{\cdot}{\Delta,Q}}{\pi_1^w}}&\deduce{\Seq{\Gamma,\neg Q, P\cimp Q}{P}{\Delta}}{\pi_2^w}}&
\deduce{\Seq{\Gamma,\neg Q, P}{\cdot}{\Delta}}{\pi_1^\equiv}}}
$
}

\noindent
where $\pi_1^\equiv$ is the same as $\pi_1$ where every application of the rule $\D$ over $Q$ in the above derivation is substituted by an application of $\neg$ over $\neg Q$. Observe that  the cut-formula of  $N-\cut_1$ has lower ecumenical weight than  $N-\cut_0$, while the cut-height of  $N-\cut_2$ is smaller than $N-\cut_0$. Finally, observe that this is a non-trivial cut-reduction: usually, the cut over the implication is replaced by a cut over $Q$ first. Due to polarities, if $Q$ is positive, then $\neg Q$ is negative and cutting over it will add to the left context the classical information $Q$, hence mimicking the behavior of formulas in the right input context.

-  $N=\cexists x. P$, with $P$ positive. 
$$
\infer[N-\cut_0]{\Seq\Gamma\cdot{\Delta}}
{\infer[\cexists R]{\Seq\Gamma\cdot{\Delta,\cexists x. P}}
{\deduce{\Seq{\Gamma}\cdot{\Delta,\cexists x. P,P[y/x]}}{\pi_1}}&
\infer[\cexists L]{\Seq{\Gamma,\cexists x. P}\cdot{\Delta}}
{\deduce{\Seq{\Gamma,P[y/x]}{\cdot}{\Delta}}{\pi_2}}}
$$
reduces to

\resizebox{\textwidth}{!}{$
\infer[N-\cut_1]{\Seq\Gamma\cdot{\Delta}}
{\infer[\neg R]{\Seq\Gamma\cdot{\Delta,\neg P[y/x]}}
{\deduce{\Seq{\Gamma, P[y/x]}\cdot{\Delta}}{\pi_2}}
&
\infer[N-\cut_2]{\Seq{\Gamma,\neg P[y/x]}{\cdot}{\Delta}}
{\deduce{\Seq{\Gamma,\neg P[y/x]}{\cdot}{\Delta,\cexists x.P}}{\pi_1^\equiv}&
\infer[\cexists L]{\Seq{\Gamma,\neg P[y/x],\cexists x.P}{\cdot}{\Delta}}
{\deduce{\Seq{\Gamma,\neg P[y/x],P[z/x]}{\cdot}{\Delta}}{\pi_2[z/y]}}}}
$}

\noindent
where the same observations for the above case hold, and $\pi_2[z/y]$ indicates the renaming of fresh variables in the derivation $\pi_2$. 
\end{itemize}
\end{proof}
We finish this section noting that polarities play an important role in the cut-elimination process. In fact, without them, adding a general cut rule would collapse the system to classical logic.
\begin{example}
If the cut rule
$$
\infer[\cut]{\Seq\Gamma\Pi\Delta}
{\Seq\Gamma{\Pi^*}{\Delta,A}&
\Seq{A,\Gamma}\Pi\Delta}
$$
was admissible in $\LCE$ for an arbitrary formula $A$, then $\Seq\cdot{A\ivee \neg A}\cdot$ would have the proof
$$
\infer[\cut]{\Seq\cdot{A\ivee \neg A}\cdot}
{\infer[\D]{\Seq\cdot\cdot{A\ivee \neg A}}
{\infer[\ivee]{\Seq\cdot{A\ivee \neg A}{A\ivee \neg A}}
{\infer[\neg R]{\Seq\cdot{\neg A}{A\ivee \neg A}}
{\infer[\D]{\Seq{A}{\cdot}{A\ivee \neg A}}
{\infer[\ivee]{\Seq{A}{A\ivee \neg A}{A\ivee \neg A}}
{\infer[\ginit]{\Seq{A}{A}{A\ivee \neg A}}{}}}}}}&
\infer[\ginit]{\Seq{A\ivee \neg A,\Gamma}{A\ivee \neg A}\cdot}{}}
$$
\end{example}

%% file: modal.tex



We will now extend the propositional fragment of $\LCE$ with modalities.

The language of \emph{(propositional, normal) modal formulas} consists of a denumerable set $\Prop$ of propositional symbols and a 
set of propositional connectives enhanced with the unary \emph{modal operators} $\square$ and $\lozenge$ concerning necessity and possibility, respectively~\cite{blackburn_rijke_venema_2001}. 

The semantics of modal logics is often determined by means of \emph{Kripke models}. Here, we will follow the approach in~\cite{Sim94}, where a modal logic is characterized by the respective interpretation of the modal model in the meta-theory (called {\em meta-logical characterization}). 

Formally, given a variable $x$, we recall the standard translation $\tradm{\cdot}{x}$ from modal formulas into first-order formulas with at most one free variable, $x$, as follows: if $p$ is atomic, then 
$\tradm{p}{x}=p(x)$; $\tradm{\bot}{x}=\bot$;  for any binary connective $\star$, $\tradm{A\star B}{x}=\tradm{A}{x}\star \tradm{B}{x}$; for the modal connectives
\begin{center}
	\begin{tabular}{lclc@{\quad}lcl}
		$\tradm{\square A}{x}$ & = & $\forall y ( \relfo(x,y) \rightarrow \tradm{A}{y})$ & &
		$\tradm{\lozenge A}{x}$ & = & $\exists y (\relfo(x,y) \wedge \tradm{A}{y})$
		\\
	\end{tabular}
\end{center}
where $\rel(x,y)$ is a binary predicate. 

Opening a parenthesis: such a translation has, as underlying justification, the interpretation of alethic modalities in a Kripke model $\m=(W,R,V)$:
\begin{equation}\label{eq:kripke}
\begin{array}{l@{\qquad}c@{\qquad}l}
\m, w \models \square A & \mbox{ iff } & \text{for all } v \text{ such that  }w \rel v, \m, v \models A.\\
\m, w \models \lozenge A & \mbox{ iff } & \text{there exists } v \text{ such that  } w \rel v	\text{ and } \m, v\models A. 
\end{array}
\end{equation}
$\rel(x,y)$ then represents the \emph{accessibility relation} $R$ in a Kripke frame. This intuition  can be made formal based on the one-to-one correspondence between classical/intuitionistic translations and Kripke modal models~\cite{Sim94}. We close this parenthesis by noting that this justification is only motivational, aiming at introducing modalities. 

The object-modal logic OL is then characterized in the first-order meta-logic ML as
$$
\vdash_{OL} A \quad\mbox{ iff } \quad \vdash_{ML} \forall x. \tradm{A}{x}
$$
Hence, if ML is classical logic (CL), the former definition characterizes the classical modal logic $\logick$~\cite{blackburn_rijke_venema_2001}, while if it is intuitionistic logic (IL), then it characterizes the intuitionistic modal logic $\logicik$~\cite{Sim94}. 

In this work, we will adopt ecumenical logic as the meta-theory (given by the system $\LCE$), hence characterizing what we will define as the {ecumenical modal logic} $\EK$. 

\subsection{An ecumenical view of modalities}\label{sec:view}
The language of \emph{ecumenical modal formulas} consists of a denumerable set $\Prop$ of (ecumenical) propositional symbols and the set of ecumenical connectives enhanced with unary \emph{ecumenical modal operators}. Unlike for the classical case, there is not a canonical definition of constructive or intuitionistic modal logics. Here we 
will mostly follow the approach in~\cite{Sim94} for justifying our choices for the ecumenical interpretation for {\em possibility} and {\em necessity}.

The ecumenical  translation $\tradme{\cdot}{x}$ from propositional ecumenical formulas into $\LCE$ is defined in the same way as the modal translation $\tradm{\cdot}{x}$ in the last section.
For the case of modal connectives, observe that, due to Proposition~\ref{prop:ev}, the interpretation of ecumenical consequence should be essentially {\em intuitionistic}. 
This implies that the box modality is a  {\em neutral connective}. The diamond, on the other hand, has two possible interpretations: classical and intuitionistic, since its leading connective is an existential quantifier. Hence we should have the ecumenical modalities: $\square, \ilozenge, \clozenge$, determined by the translations
\vspace{-.1cm}
\begin{center}
	\begin{tabular}{lclc@{\qquad}lcl}
	$\tradme{\square A}{x}$ & = & $\forall y ( \relfo(x,y) \iimp \tradme{A}{y})$\\
		$\tradme{\ilozenge A}{x}$ & = & $\iexists y (\relfo(x,y) \wedge \tradme{A}{y})$ & &
		$\tradme{\clozenge A}{x}$ & = & $\cexists y (\relfo(x,y) \wedge \tradme{A}{y})$
		\\
	\end{tabular}
\end{center}
We will denote by $\EK$ the ecumenical modal logic meta-logically characterized by  $\LCE$ via $\tradme{\cdot}{x}$.
Polarities will be extended to the modal case smoothly, that is, formulas with outermost connective classical or negation are negative, all the others are positive. Relational atoms are not polarizable.

In Figure~\ref{fig:labEK} we present the pure, labeled ecumenical system $\labEK$. 
Observe that
$$\vdash_{\labEK} x:\clozenge A\iequiv x:\neg\square\neg A
$$
On the other hand, $\square$ and $\ilozenge$ are not inter-definable.
Finally, if $A^{ec}$ is eventually externally classical, then
$$
\vdash_{\labEK} x:\square A^{ec}\iequiv x:\neg\clozenge\neg A^{ec}
$$
This means that, when restricted to the classical fragment, $\square$ and $\clozenge$ are duals. 
This reflects well the ecumenical nature of the defined modalities.

\begin{figure}[htp]
{\sc Intuitionistic and neutral Rules}
$$
\infer[{\wedge L}]{\Seq{\Gamma, x:A \wedge B}{ \Pi} {\Delta}}{\Seq{\Gamma,
x:A, x:B}{ \Pi}{\Delta}} 
\quad
\infer[{\wedge R}]{\Seq{\Gamma}
{x:A \wedge B}{\Delta}}{\deduce{\Seq{\Gamma} {x:A}{\Delta}}{} &\deduce{\Seq{\Gamma}{x:B}{\Delta}}{}} 
$$
$$
\infer[{\vee_i L}]{\Seq{\Gamma, x:A \vee_i B}{ \Pi}{\Delta}
}{\deduce{\Seq{\Gamma, x:A}{ \Pi}{\Delta}}{} &\deduce{\Seq{\Gamma,x:B}{\Pi}{\Delta}}{}} 
\quad
\infer[{\vee_i R_j}]{\Seq{\Gamma}{ x:A_1\vee_i A_2}{\Delta}}{\Seq{\Gamma}{ x:A_j}{\Delta}} 
$$
$$
\infer[{\iimp L}]{\Seq{\Gamma, x:A\iimp B}{
\Pi}{\Delta}}
{\Seq{\Gamma, x:A\iimp B}{ x:A}{\Delta} & \Seq{\Gamma , x:B} 
{\Pi}\Delta} 
\quad
\infer[{\iimp R}]{\Seq{\Gamma}
{x:A\iimp B}{\Delta}}{\Seq{\Gamma,x:A}{x:B}{\Delta}}
$$
$$
\infer[{\neg L}]{\Seq{\Gamma,x:\neg A}{\cdot}{\Delta}
}{\Seq{\Gamma, x:\neg A}{  x:A}{\Delta}}
\quad
\infer[{\neg R}]{\Seq{\Gamma} 
{\cdot}{x:\neg A,\Delta}}{\Seq{\Gamma,x:A}{\cdot}{\Delta}}
$$
{\sc Classical Rules}
$$
\infer[\cimp L]{\Seq{\Gamma, x:A \rightarrow_c B}{\cdot}{\Delta} }{\Seq{\Gamma, x:A \rightarrow_c B}{x:A}{\Delta}&\Seq{\Gamma, x: B}{\cdot}{\Delta}}
\quad
\infer[\cimp R]{\Seq{\Gamma} {\cdot}{ x:A\rightarrow_c B,\Delta}}{\Seq{\Gamma, x:A}{  \cdot}{x:B,\Delta}} 
$$
$$
\infer[\vee_cL]{\Seq{\Gamma,  x: A \vee_c B}{\cdot}{\Delta}}{\Seq{\Gamma,x:A}{ \cdot}{\Delta} & \Seq{\Gamma,x:B}{ \cdot}{\Delta}} 
\quad
\infer[\vee_cR]{\Seq{\Gamma}{\cdot}{x:A \vee_c B, \Delta}}{\Seq{\Gamma}{ \cdot}{x:A,x:B,\Delta}}
$$
$$
\infer[L_c]{\Seq{\Gamma, x:p_c}{\cdot}{\Delta}}{\Seq{\Gamma,x:p_i}{ \cdot}{\Delta}}
\quad
\infer[R_c]{\Seq{\Gamma}{\cdot}{x:p_c,\Delta}}{\Seq{\Gamma}{ \cdot}{x:p_i,\Delta}} 
$$
{\sc Modal rules}
$$
	\infer[\square L]{\Seq{xRy, x:\Box A,\Gamma}{\Pi}{\Delta}}{\Seq{xRy,y:A,x:\Box A,\Gamma}{\Pi}{\Delta}}
	\quad
	\infer[\square R]{\Seq{\Gamma}{x:\Box A}{\Delta}}{\Seq{xRy, \Gamma}{y:A}{\Delta}}
	\quad
	\infer[\ilozenge L]{\Seq{x:\ilozenge A,\Gamma}{\Pi}{\Delta}}{\Seq{xRy,  y: A,\Gamma}{\Pi}{\Delta}}
	$$
	$$
	\infer[\ilozenge R]{\Seq{xRy, \Gamma}{x:\ilozenge A}{\Delta}}{\Seq{xRy,\Gamma}{y:A}{\Delta}}
	\quad
	\infer[\clozenge L]{\Seq{x:\clozenge A,\Gamma}{\cdot}{\Delta} }{\Seq{xRy,  y: A,\Gamma}{\cdot}{\Delta}}
	\quad
	\infer[\clozenge R]{\Seq{xRy,\Gamma}{\cdot}{x:\clozenge A,\Delta}}{\Seq{xRy,\Gamma}{\cdot}{y:A,x:\clozenge A,\Delta}}
	$$
{\sc Initial, Decision and Structural Rules}
$$
\infer[\init_i]{\Seq{\Gamma,x:A}{x:A}{\Delta}}{} 
\qquad
\infer[\init_c]{\Seq{\Gamma,x:A}{\Pi}{x:A,\Delta}}{} 
$$
$$
\infer[\D]{\Seq{\Gamma}{ \cdot}{x:P,\Delta}}{\Seq{\Gamma}{ x:P}{x:P,\Delta}} 
\qquad
\infer[\store]{\Seq{\Gamma}{x:N}{\Delta}}{\Seq{\Gamma}{ \cdot}{x:N,\Delta}}
\qquad
\infer[\W]{\Seq{\Gamma}{x:A}{\Delta}}{\Seq{\Gamma}{ \cdot}{\Delta}} 
$$
{\sc Cut Rules}
$$
\begin{array}{c@{\qquad}c}
\infer[P-\cut]{\Seq{\Gamma}{ \Pi}{\Delta}}{\Seq{\Gamma}{x:P}{\Delta} & \Seq{x:P,\Gamma}{ \Pi}{\Delta}} 
&
\infer[N-\cut]{\Seq{\Gamma}{ \Pi}{\Delta}}{\Seq{\Gamma}{ \Pi^*}{\Delta, x:N} & \Seq{x:N,\Gamma}{ \Pi}{\Delta}} 
\end{array}
$$\caption{Ecumenical pure modal system $\labEK$. In rules $\square R,\ilozenge L, \clozenge L$,  the eigenvariable $y$ does not occur free in any formula of the conclusion; $N$ is negative and $P$ is positive; $p$ is atomic; $\Pi^*$ is either empty or some $z:P\in\Delta$.}\label{fig:labEK}
\end{figure}

\subsection{Ecumenical birelational models}\label{sec:bi}
The ecumenical birelational Kripke semantics, which is an extension of the proposal in~\cite{luiz17} to modalities, was presented in~\cite{DBLP:conf/dali/MarinPPS20}.

\begin{definition}\label{def:ekripke} A {\em birelational Kripke model} is a quadruple  
$\m=(W,\leq,R,V)$ where $(W,R,V)$ is a Kripke model
such that $\wld$ is partially ordered with order $\leq$, $R\subset W\times W$ is a binary relation, the 
satisfaction function $V:\langle W,\leq\rangle\rightarrow  \langle2^\Prop,\subseteq\rangle$ is monotone and:

\noindent
F1. For all worlds $w,v,v'$, if $wRv$ and $v \leq v'$, there is a $w'$ such that $w \leq w'$ and $w'Rv'$;

\noindent
F2. For all worlds $w', w, v$, if $w \leq w'$ and $wRv$, there is a $v'$ such that $w'Rv'$ and $v \leq v'$.

An {\em ecumenical modal Kripke model} is a birelational Kripke model such that truth of an ecumenical formula at a point $w$ is the
smallest relation $\modelse$ satisfying 

\noindent
$
\begin{array}{l@{\quad}c@{\quad}l}
\m,w\modelse p_{i} & \mbox{ iff } & p_{i}\in \val(w);\\
\m,w\modelse A\wedge B & \mbox{ iff } & \m,w\modelse A \mbox{ and } \m,w\modelse B;\\
\m,w\modelse A\vee_i B & \mbox{ iff } & \m,w\modelse A \mbox{ or } \m,w\modelse B;\\
\m,w\modelse A\iimp B & \mbox{ iff } & \text{for all } v \text{ such that  }w \leq v,   
\m,v\modelse A
\mbox{ implies } 
\m,v\modelse B;\\
\m,w\modelse \neg A & \mbox{ iff } & \text{for all } v \text{ such that  }w \leq v,   
\m,v\not\modelse A;\\
\m,w \modelse \bot & & \mbox{never holds};\\
\m, w \modelse \square A & \mbox{ iff } & \text{for all }  v,w' \text{ such that  }w\leq w'\text{ and }  w'\rel v, \m, v \modelse A.\\
\m, w \modelse \ilozenge A & \mbox{ iff } & \text{there exists } v \text{ such that  } w \rel v	\text{ and } \m, v\modelse A. \\
\m,w\modelse p_c  & \mbox{ iff } & \m,w\modelse \neg(\neg p_i);\\
\m,w\modelse A\vee_c B & \mbox{ iff } & \m,w\modelse \neg(\neg A\wedge \neg B);\\
\m,w\modelse A\cimp B & \mbox{ iff } & \m,w\modelse  \neg(A\wedge \neg B).\\
\m, w \modelse \clozenge A & \mbox{ iff } & \m,w\modelse  \neg\square\neg A.
\end{array}
$

We say that a formula $A$ is {\em valid} in a model $\m=(W,\leq,R,V)$ if for all $w\in  W$ we have $w\modelse A$. A formula $A$ is {\em valid} in a frame $\langle W,\leq,R\rangle$ if, for all valuations $V$, $A$ is valid in the model $(W,\leq,R,V)$. Finally, we say a formula is {\em valid}, if it is valid in all frames.
\end{definition}
Since, restricted to intuitionistic and neutral connectives, $\modelse$ is the usual birelational interpretation $\models$ for $\IK$~\cite{Sim94}, and since the classical connectives are interpreted 
via the neutral ones using the double-negation translation, an ecumenical modal Kripke model coincides with the standard birelational Kripke model for intuitionistic modal logic $\IK$. Hence the following result easily holds from the similar result for $\IK$.
\begin{theorem}[\cite{DBLP:conf/dali/MarinPPS20}]\label{thm:scLEci} The system
$\labEK$ is sound and complete w.r.t.~the ecumenical modal Kripke semantics, that is, $\vdash_{\labEK} x: A$ iff $\modelse A$.
\end{theorem}
We end this section with a small note on the relationship between the semantics and the dynamics of proofs. 
On a bottom-up reading of proofs, the $\store$ rule is a delay on applying rules over classical connectives. It corresponds to  moving the formula up w.r.t. $\leq$ in the birelational semantics. The rule $\Box R$, on the other hand, slides the formula to a fresh new world, related to the former one through the relation $R$. Finally, rule $\neg R$ moves up the formula w.r.t. $\leq$.

%% file: nested.tex

The  criticism regarding system $\labEK$ is that
it includes {\em labels} in the technical machinery, hence 
allowing one to write sequents that cannot always be interpreted within the ecumenical modal language.

This section is devoted to present the pure {\em label free} calculus for ecumenical modalities introduced in~\cite{DBLP:conf/wollic/MarinPPS21}, where every basic object of the calculus can be translated as a formula in the language of the logic. 

The inspiration comes from  Stra{\ss}burger's {\em nested system} for $\IK$~\cite{DBLP:conf/fossacs/Strassburger13}. The main idea is to add nested layers to sequents, which intuitively corresponds to worlds in a relational structure~\cite{Fitting:2014,Brunnler:2009kx,Poggiolesi:2009vn}.

The  structure of a nested
sequent for ecumenical modal logics is a
tree whose nodes are multisets of formulas, just like in~\cite{DBLP:conf/fossacs/Strassburger13}, with the relationship between parent and child in the tree represented by bracketing $\BR{\cdot}$.
The difference however is that the
ecumenical formulas can be {\em left inputs} (in the left contexts -- marked with a full circle $\lf{}$), {\em right inputs} (in the classical right contexts -- marked with a triangle $\ct{}$) or a {\em single right output} (the stoup -- marked with a white circle $\rt{}$).  

\begin{definition}
 \emph{Ecumenical nested sequents} are defined in terms of a grammar of \emph{input sequents} 
 (written $\Lx$) and \emph{full sequents} (written $\Rx$) where the left/right input formulas are denoted by $\lf{A}$ and $\ct{A}$, respectively,  and $\rt{A}$ denote the output formula. When the distinction between input and full sequents is not essential or cannot be made explicit, we will use $\Dx$ to stand for either case. 
	\begin{align*}
	\Lx &\coloneqq \left. 
	\emptyset \mid 
	\lf{A}, \Lx \mid \ct{A}, \Lx\mid  \BR{\Lx}\right. &
	\Rx &\coloneqq \rt{A},\Lx \mid [\Rx],\Lx &
	\Dx &\coloneqq \Lx\mid \Rx
	\end{align*}
As usual, we allow sequents to be empty, 
and we consider sequents to be equal modulo associativity and commutativity of the comma. 

We write 
$\Rxb$ for the result of replacing an output formula 
from $\Rx$
by $\rt{\bot}$,
while  
$\Lxb$ represents the result of adding anywhere of the input context $\Lx$ the output formula  $\rt{\bot}$. 
Finally, 
$\Dx*$ is the result of erasing an output formula (if any) from 
$\Dx$.
\end{definition}

Observe that full sequents $\Rx$ necessarily contain exactly one output-like formula, having the form
$\Lx[1],\BR{\Lx[2],\BR{ \ldots,\BR{\Lx[n],\rt{A}}}\ldots}$.

\begin{example}
The nested sequent $\ct{\clozenge A},\BR{\rt{\neg B}}, \BR{\lf{C\wedge D}}$ represents the following tree of sequents 
\begin{center}
\begin{tikzpicture}[level distance=3em, 
			level 1/.style={sibling distance=6em},
			every node/.style = {align=center}]]
			\node {$\Seq\cdot\cdot{\clozenge A}$}
			child[thick] { node {$\Seq\cdot{\neg B}{\cdot}$}}
			child[thick] { node {$\Seq{C\wedge D}\cdot{\cdot}$}};
		\end{tikzpicture}
\end{center}
\end{example}

The next definition (of contexts) allows for identifying subtrees within nested sequents, which is necessary for introducing inference rules in this setting.
\begin{definition}
	An \emph{$n$-ary context} $\Dx\i{\scriptscriptstyle 1}{}\dots\i{n}{}$ is like a sequent but contains $n$ pairwise distinct numbered \emph{holes} $\Ex{}$ wherever a formula may otherwise occur. 	
	It is a \emph{full} or a \emph{input} context when $\Dx = \Rx$ or $\Lx$ respectively.
		
Given $n$ sequents $\Dx[1], \dotsc, \Dx[n]$, we write $\Dx{\Dx[1]}\dots{\Dx[n]}$  for the sequent where the i-th hole in $\Dx\i{\scriptscriptstyle 1}{}\dots\i{n}{}$ has been replaced by $\Dx[i]$ (for $1 \le i \le n$), assuming that the result is well-formed, \emph{i.e.}, there is at most one output formula. If $\Dx[i]=\emptyset$ the  hole is removed.

Given two nested contexts $\Rx^i\{\}=\Dx[1]^i,\BR{\Dx[2]^i,\BR{ \ldots,\BR{\Dx[n]^i,\{ \}}}\ldots}$, $i\in\{1,2\}$, their {\em merge}\footnote{As observed in~\cite{Poggiolesi:2009vn,DBLP:conf/tableaux/Lellmann19}, the merge is a ``zipping" of the two nested sequents along the path from the root to the hole.} is
\[
\Rx^1\otimes\Rx^2\{\}=\Dx[1]^1,\Dx[1]^2,\BR{\Dx[2]^1,\Dx[2]^1,\BR{ \ldots,\BR{\Dx[n]^1,\Dx[n]^2,\{ \}}}\ldots}
\]
\end{definition}
Figure~\ref{fig:nEK} presents the nested sequent system $\nEK$ for ecumenical modal logic $\EK$. 

\begin{figure}[t]
{\sc Intuitionistic and neutral Rules}
$$
\infer[\lf\wedge]{\Rx{\lf{A \wedge B}}}{
	\Rx{\lf A, \lf B}
}
\quad
\infer[\rt\wedge]{\Lx{\rt{A \wedge B}}}{
	\Lx{\rt A}
	&
	\Lx{\rt B}
}
\quad
\infer[\lf\ivee]{\Rx{\lf{A \ivee B}}}{
	\Rx{\lf A}
	&
	\Rx{\lf B}
}
\quad
\infer[\rt\vee_{i_j}]{\Lx{\rt{A_1 \ivee A_2}}}{
	\Lx{\rt A_j}
}
\quad 
\infer[\lf\bot]{\Rx{\lf\bot}}{}
$$
$$
\infer[\lf\iimp]{\Rx{\lf{A \iimp B}}}{
	\Rx*{\lf{A \iimp B},\rt A}
	&
	\Rx{\lf B}
}
\quad
\infer[\rt\iimp]{\Lx{\rt{A \iimp B}}}{
	\Lx{\lf A, \rt B}
}
\quad
\infer[\lf\neg]{\Rxb{\lf{\neg A}}}{
	\Rx*{\lf{\neg A}, \rt A}
}
\quad
\infer[\ct\neg]{\Rxb{\ct{\neg A}}}{
	\Rxb{\lf A}
}
$$
{\sc Classical Rules}
$$
\infer[\lf\cimp]{\Rxb{\lf{A \cimp B}}}{
	\Rx*{\lf{A \cimp B}, \rt A}
	&
	\Rxb{\lf B}}
\quad
\infer[\ct\cimp]{\Rxb{\ct{A \cimp B}}}{
	\Rxb{\lf A, \ct{B}}
}
\quad
\infer[\lf\cvee]{\Rxb{\lf{A \cvee B}}}{
	\Rxb{\lf A}
	&
	\Rxb{\lf B}
}
$$
$$
\infer[\ct\cvee]{\Rxb{\ct{A \cvee B}}}{
	\Rxb{\ct{A}, \ct{B}}
}
\quad
\infer[\lf p_{c}]{\Rxb{\lf{p_c}}}{
	\Rxb{\lf{p_i}}}
\quad
\infer[\ct p_{c}]{\Rxb{\ct{p_c}}}{
	\Rxb{\ct{p_i}}}
$$
{\sc Modal rules}
$$
\infer[\lf\BOX]{\Dx[1]{\lf{\BOX A}, \BR{\Dx[2]}}}{ 
	\Dx[1]{\lf{\BOX A}, \BR{\lf A, \Dx[2]}}
}
\quad
\infer[\rt\BOX]{\Lx{\rt{\BOX A}}}{
	\Lx{\BR{\rt A}}
}
\quad
\infer[\lf\ilozenge]{\Rx{\lf{\ilozenge A}}}{
	\Rx{\BR{\lf A}}
}
\quad
\infer[\rt\ilozenge]{\Lx[1]{\rt{\ilozenge A}, \BR{\Lx[2]}}}{
	\Lx[1]{\BR{\rt A, \Lx[2]}}
}
$$
$$
\infer[\lf\clozenge]{\Rxb{\lf{\clozenge A}}}{
	\Rxb{\BR{\lf A}}}
\quad
\infer[\ct\clozenge]{\Dxb[1]{\ct{\clozenge A},\BR{\Dxb[2]}}}{
	\Dxb[1]{\ct{\clozenge A},\BR{\ct{A},\Dxb[2]}}
}
$$
{\sc Initial, Decision and Structural Rules}
$$
\infer[\ginit]{\Lx{\lf{A},\rt{A}}}{} 
\quad 
\infer[\gcinit]{\Rxb{\lf{A},\ct{A}}}{}  
\quad 
\infer[\D]{\Rxb{\ct{P}}}{
	\Rx*{\ct{P},\rt{P}}
} 
\quad
\infer[\store]{\Lx{\rt{N}}}{
	\Lx{\ct{N},\rt{\bot}}
} 
\quad
\infer[\W]{\Rx}{\Rxb}
$$
{\sc Cut Rules}
$$
\infer[\icut]{\Rx{\emptyset}}
{\Rx*{\rt{P}}&\Rx{\lf{P}}}\qquad
\infer[\ccut]{\Rx{\emptyset}}
{\RxP{\ct{N}}&\Rx{\lf{N}}}
$$
\caption{Nested ecumenical modal system $\n\EK$. $P$ is a {\em positive} formula, $N$ is a {\em negative} formula. $p$ is atomic. $\Gamma^P$ denotes either $\Gamma^{\rt{\bot}}$ or $\Gamma^*\{\ct{P},\rt{P}\}$ for some $\ct{P}\in\Gamma$.}\label{fig:nEK}\end{figure}

\begin{example}\label{ex:nes} Below right is the proof that $\clozenge A$ is a consequence of $\neg\Box\neg A$ for any formula $A$. Below left the proof that, if $N$ is negative, then $\Box N$ is a consequence of $\neg\ilozenge\neg N$. In fact, this holds for and only for eventually externally classical formulas (see Definition~\ref{def:eec}).
$$
\infer[\rt{\iimp}]{\rt{\neg\Box\neg A\iimp\clozenge A}}
{\infer[\store]{\lf{\neg\Box\neg A},\rt{\clozenge A}}
{\infer[\lf\neg]{\lf{\neg\Box\neg A},\ct{\clozenge A},\rt\bot}
{\infer[\rt\Box]{\rt{\Box\neg A},\ct{\clozenge A}}
{\infer[\ct\clozenge]{\ct{\clozenge A},\BR{\rt{\neg A}}}
{\infer[\rt\neg]{\BR{\ct{A},\rt{\neg A}}}
{\infer[\gcinit]{\BR{\lf{A},\ct{A},\rt\bot}}{}}}}}}}
\qquad\qquad\qquad
\infer[\rt\iimp]{\rt{\neg \ilozenge\neg N\iimp\Box N}}
{\infer[\rt{\Box}]{\lf{\neg \ilozenge\neg N},\rt{\Box N}}
{\infer[\D]{\lf{\neg \ilozenge\neg N},\BR{\rt{N}}}
{\infer[\lf{\neg}]{\lf{\neg \ilozenge\neg N},\BR{\ct{N},\rt{\bot}}}
{\infer[\rt{\clozenge}]{\rt{\ilozenge\neg N},\BR{\ct{N}}}
{\infer[\store]{\BR{\rt{\neg N},\ct{N}}}
{\infer[\ct{\neg}]{\BR{\ct{\neg N},\ct{N}}}
{\infer[\gcinit]{\BR{\lf{N},\ct{N}}}{}}}}}}}}
$$
\end{example}

\subsection{Proof theoretic properties}\label{sec:cut}
As for $\labEK$, the properties of $\nEK$ are inherited by the ones in $\LCE$ (see Lemma~\ref{lemma:wc}). We will list them explicitly since the notation is quite different.
\begin{theorem}\label{thm:pt}
In $\nEK$:
	\begin{enumerate}
		\item The rules $\lf\cvee, \ct\cvee, \lf\cimp,\ct\cimp,\lf\neg,\rt\neg,\lf{p_c},\ct{p_c},\lf\clozenge,\ct\clozenge$ and $\D$ are invertible.
		
		\item The rules $\lf\wedge, \rt\wedge, \lf\ivee, \rt\iimp,\lf\ilozenge, \lf\Box,\rt\Box$ and $\store$ are  totally invertible.		
		\item 
		The following structural
		rules are admissible 
		$$
		\infer[\W_c]{\Lx\otimes\Rx}{\Rx} \qquad 
		\infer[\C_c]{\Lx\otimes\Rx}{\Lx\otimes\Lx\otimes\Rx}
		$$
	\item The rules $\icut$ and $\ccut$ are admissible. The ecumenical weight is the following extension of the measure presented in Section~\ref{sub:cut} for propositional connectives 
\begin{center}
\begin{tabular}{lc@{\qquad}l}
$\ew(\heartsuit A)=\ew(A)+1$ if $\heartsuit\in\{\ilozenge,\Box\}$ & & 
$\ew(\clozenge A)=\ew(A)+4$ 
\end{tabular}
\end{center}

	\end{enumerate}
\end{theorem}

The invertible but not totally invertible rules in $\nEK$ concern negative formulas, hence they can only be applied in the presence of  empty stoups ($\rt\bot$).  Note also that the rules $\W,\rt\ivee,$ and $\rt\ilozenge$ are not invertible, while $\lf\iimp$ is invertible only w.r.t. the right premise.

%% file: sound-comp.tex

In this section we will show that all rules presented in Figure~\ref{fig:nEK} are  sound and complete w.r.t. the ecumenical birelational model. 
The idea is to prove that the rules of the system $\nEK$ preserve {\em validity}, in the sense that if the interpretation of the premises is valid, so is the interpretation of the conclusion. 

The first step is to determine the interpretation of ecumenical nested sequents. 
In this section, we will 
present the translation of nestings to labeled sequents, hence establishing, at the same time, soundness and completeness of $\nEK$ and the relation between this system with $\labEK$. 



\begin{definition}
Let $\lf\Sigma,\ct\Sigma,\rt\Px$ represent that all formulas in the each set/multiset are respectively input left, right, or output formulas. 
The underlying set/multiset will represent in all cases the corresponding multiset of unmarked formulas.
The translation 
$\tradm{\cdot}{x}$ from nested into labeled sequents is defined recursively by
\[
\begin{array}{rcl}
\tradm{\lf{\Sigma_1},\ct{\Sigma_2},\rt{\Px[3]},\BR{\Dx[1]},\ldots, \BR{\Dx[n]}}{x}&\coloneqq&
\left(\{xRx_i\}_i, x:\Sigma_1\seq x:\Sigma_2 ;x:\Px[3]\right)\otimes
\left\{\tradm{\Dx[i]}{x_i}
\right\}_i
\end{array}
\]
where $1\leq i\leq n$, $x_i$ are fresh, $\bot$ is translated to the empty set, and the merge operation on labeled sequents is defined as 
\[\begin{array}{rcl}
(\Gamma_1\seq\Delta_1;\Pi_1)\otimes(\Gamma_2\seq\Delta_2;\Pi_2) &\;\coloneqq\;&
\Gamma_1,\Gamma_2\seq\Delta_1,\Delta_2;\Pi_1,\Pi_2
\end{array}
\]
Since full nested sequents have exactly one output formula (which can be $\bot$), the stoup in the labeled setting
will have at most one formula, and the merge above is well defined.
Given $\mathcal{R}$ a set of relational  formulas, we will denote by $xR^*z$ the fact that there is a path from $x$ to $z$ in $\mathcal{R}$, i.e., there are $y_j\in\mathcal{R}$ for $0\leq j\leq k$ such that $x=y_0, y_{j-1}Ry_j$ and $y_k=z$.
\end{definition}

\begin{theorem}\label{theo:trans}
Let $\Gamma$ be a nested sequent and $x$ be any label. The following are equivalent.
\begin{enumerate}
\item  $\Gamma$ is provable in $\nEK$;
\item $\tradm{\Gamma}{x}$ is provable in $\labEK$.
\end{enumerate}
\end{theorem}
\begin{proof} 
Let $xR^*z\in\mathcal{R}$. Observe that:
\begin{itemize}
\item[-] $\tradm{\Rxb{\lf{\clozenge A}}}{x}=\mathcal{R},\Sigma,z:\clozenge A \seq \Delta; x:\bot$ iff\\  $\tradm{\Rxb{\BR{\lf A}}}{x}=\mathcal{R},zRy,\Sigma, y: A \seq \Delta; x:\bot$, with $y$ fresh.
\item[-] $\tradm{\Dxb[1]{\ct{\clozenge A},\BR{\Dxb[2]}}}{x}=\mathcal{R},zRy,\Sigma \seq \Delta,z: \clozenge A; x:\bot$
with $y$ the variable related to the nesting of $\Delta_2$ iff\\
$\tradm{\Dxb[1]{\ct{\clozenge A},\BR{\ct{A},\Dxb[2]}}}{x}=\mathcal{R},zRy,\Sigma \seq \Delta,z: \clozenge A,y: A; x:\bot$.
\item[-] etc.
\end{itemize}
The translation $\tradm{\cdot}{x}$ is then trivially lifted to rule applications. We will illustrate the $\clozenge$ cases.
\begin{itemize}
\item[-] Case $\lf\clozenge$.  
$$
\vcenter{\infer[\lf\clozenge]{\Rxb{\lf{\clozenge A}}}{\Rxb{\BR{\lf A}}}}
\quad\leftrightsquigarrow\quad
\vcenter{\infer[\clozenge L]{\tradm{\Rxb{\lf{\clozenge A}}}{x}}{\tradm{\Rxb{\BR{\lf A}}}{x}}}
$$
\item[-] Case $\rt\clozenge$. 
$$
\vcenter{\infer[\ct\clozenge]{\Dxb[1]{\ct{\clozenge A},\BR{\Dxb[2]}}}{
	\Dxb[1]{\ct{\clozenge A},\BR{\ct{A},\Dxb[2]}}}}
\quad\leftrightsquigarrow\quad
\vcenter{\infer[\clozenge R]{\tradm{\Dxb[1]{\ct{\clozenge A},\BR{\Dxb[2]}}}{x}}{
	\tradm{\Dxb[1]{\ct{\clozenge A},\BR{\ct{A},\Dxb[2]}}}{x}}}
$$
\end{itemize}
Given this transformation, (1) $\Leftrightarrow$ (2) is easily proved by induction on a proof of $\Gamma$/$\tradm{\Gamma}{x}$ in $\nEK$/$\labEK$.
\end{proof}

Theorems~\ref{thm:scLEci} and~\ref{theo:trans} immediately imply the following. 
\begin{corollary}
Nested system $\nEK$ is sound w.r.t. ecumenical birelational semantics.
\end{corollary}
We observe that, often, passing
from labeled to nested sequents is not a simple task, sometimes even impossible. 
In fact, although the relational atoms of a sequent appearing in $\labEK$ \emph{proofs} can be arranged so as to correspond to nestings, as shown here, if the relational context is not tree-like~\cite{DBLP:conf/aiml/GoreR12}, the existence of such a translation is not clear. For instance, how should the sequent $xRy, yRx, x:A \seq y:B$ be interpreted in modal systems with symmetrical relations?

Also thanks to their tree shape, it is possible to interpret nested sequents as ecumenical modal formulas, and hence prove soundness in the same way as in~\cite{DBLP:conf/fossacs/Strassburger13}. This direct interpretation of nested sequents as ecumenical formulas means that $\nEK$ is a so-called {\em internal proof system}. 

We  end this section by briefly showing an alternative way of  proving of soundness of $\nEK$  w.r.t. the ecumenical birelational semantics. Please refer to~\cite{DBLP:conf/wollic/MarinPPS21} for a more detailed presentation.
\begin{definition}
The  formula translation $\fm{\cdot}$ for ecumenical nested sequents is given by
\begin{center}
$
	\begin{array}{rclc@{\qquad}rcl}
	\fm{\emptyset} &\coloneqq& \TOP& &
	\fm{\lf{A}, \Lx}& \coloneqq& A \wedge \fm{\Lx}\\
	\fm{\ct{A}, \Lx}& \coloneqq& \neg A \wedge \fm{\Lx} & &
	\fm{\BR{\Lx[1]},\Lx[2]}& \coloneqq& \ilozenge\fm{\Lx[1]}\wedge\fm{\Lx[2]}\\
	\fm{\Lx, \rt{A}} &\coloneqq& \fm{\Lx} \iimp A & &
	\fm{\Lx, \BR{\Rx}} &\coloneqq& \fm{\Lx} \iimp \BOX \fm{\Rx}
	\end{array}
$
\end{center}
where all occurrences of $A \wedge \TOP$ and $\TOP \iimp A$ are simplified to $A$. We say a sequent is {\em valid} if its corresponding formula is valid.
\end{definition}


The next theorem shows that the rules of $\n\EK$ preserve validity in ecumenical modal frames w.r.t. the formula interpretation $\fm{\cdot}$.
\begin{theorem}
Let  
$$
\vcenter{\infer[r]{\Gamma}{\Gamma_1\quad\ldots\quad\Gamma_n}} \quad n\in\{0,1,2\}
$$
be an instance of the rule $r$ in the system $\n\EK$.
Then $\fm{\Gamma_1}\wedge\ldots\wedge\fm{\Gamma_n}\iimp\fm{\Gamma}$ is valid in the birelational ecumenical semantics.
\end{theorem}

%% file: fragments.tex

In this section, we discuss fragments, axioms and extensions of $\nEK$.  

\subsection{Extracting fragments}
 For the sake of simplicity, in this sub-section negation will not be considered a primitive connective, it will rather take its respective intuitionistic or classical form.
\begin{definition}
An ecumenical modal formula $C$ is {\em classical} (resp. {\em intuitionistic}) if it is built from classical (resp. intuitionistic) atomic propositions using only neutral and classical (resp. intuitionistic) connectives but negation, which will be replaced by $A\cimp \bot$ (resp. $A\iimp \bot$). 
\end{definition}
The first thing to observe is that, when only pure fragments are concerned, weakening is admissible (remember that this is not the case for the whole system $\nEK$ -- see Example~\ref{ex:weak}).  Also,  only positive (resp. eventually externally classical) formulas are present in the intuitionistic (resp. classical) fragment.  

Let $\nEK_i$ (resp. $\nEK_c$) be the system obtained from $\nEK-\W$
by restricting the rules to the intuitionistic (resp. classical) case -- see Figures~\ref{fig:nEKi} and~\ref{fig:nEKc}.

The intuitionistic fragment does not have classical input formulas and it coincides with the system $\NIK$  in~\cite{DBLP:conf/fossacs/Strassburger13}. 

%
%

Regarding $\nEK_c$, since all the classical/neutral rules are invertible,
%
the following proof strategy is complete: 
\begin{enumerate}[i.]
\item Apply the rules $\lf\wedge,\rt\wedge, \lf\Box, \rt\Box$ and $\store$ eagerly, obtaining leaves of the form $\Lx{\rt{\bot}}$. 
\item Apply any other rule of $\nEKc$ eagerly, until either finishing the proof with an axiom application or obtaining leaves of the form $\Lx{\rt{P}}$, where $P$ is a positive formula in $\nEKc$, that is, having as main connective $\wedge$ or $\Box$. Start again from step (i).
\end{enumerate}
This discipline corresponds to the focused strategy for a fragment of the two-sided version of the polarized system defined in~\cite{DBLP:conf/rta/ChaudhuriMS16}, exchanging the polarities of diamond and box (which, as observed in the {\em op.cit.}, is a matter of choice since all rules are invertible).

\begin{figure}[t]
	\[
	\infer[\init]{\Lx{\lf{p_i},\rt{p_i}}}{} 
	\quad 
	\infer[\lf\bot]{\Rx{\lf\bot}}{}
	\quad
	\infer[\lf\wedge]{\Rx{\lf{A \wedge B}}}{
		\Rx{\lf A, \lf B}
	}
	\quad
	\infer[\rt\wedge]{\Lx{\rt{A \wedge B}}}{
		\Lx{\rt A}
		&
		\Lx{\rt B}
	}
	\]
	\[
	\infer[\lf\ivee]{\Rx{\lf{A \ivee B}}}{
		\Rx{\lf A}
		&
		\Rx{\lf B}
	}
	\quad
	\infer[\rt\vee_{i_j}]{\Lx{\rt{A_1 \ivee A_2}}}{
		\Lx{\rt A_j}
	}
	\quad
	\infer[\lf\iimp]{\Rx{\lf{A \iimp B}}}{
		\Rx*{\lf{A \iimp B},\rt A}
		&
		\Rx{\lf B}
	}
	\quad
	\infer[\rt\iimp]{\Lx{\rt{A \iimp B}}}{
		\Lx{\lf A, \rt B}
	}
	\]
	\[
	\infer[\lf\BOX]{\Dx[1]{\lf{\BOX A}, \BR{\Dx[2]}}}{ 
		\Dx[1]{\lf{\BOX A}, \BR{\lf A, \Dx[2]}}
	}
	\quad
	\infer[\rt\BOX]{\Lx{\rt{\BOX A}}}{
		\Lx{\BR{\rt A}}
	}
	\quad
	\infer[\lf\ilozenge]{\Rx{\lf{\ilozenge A}}}{
		\Rx{\BR{\lf A}}
	}
	\quad
	\infer[\rt\ilozenge]{\Lx[1]{\rt{\ilozenge A}, \BR{\Lx[2]}}}{
		\Lx[1]{\BR{\rt A, \Lx[2]}}
	}
	\]
	\caption{Intuitionistic fragment $\n\EK_i$.}\label{fig:nEKi}\end{figure}
\begin{figure}[t]
	\[
	\infer[\init]{\Rx{\lf{p_c},\ct{p_c}}}{} 
	\quad 
	\infer[\lf\bot]{\Rx{\lf\bot}}{}
	\quad
	\infer[\lf\wedge]{\Rx{\lf{A \wedge B}}}{
		\Rx{\lf A, \lf B}
	}
	\quad
	\infer[\rt\wedge]{\Lx{\rt{A \wedge B}}}{
		\Lx{\rt A}
		&
		\Lx{\rt B}
	}
	\]
	\[
	\infer[\lf\cvee]{\Rxb{\lf{A \cvee B}}}{
		\Rxb{\lf A}
		&
		\Rxb{\lf B}
	}
	\quad
	\infer[\ct\cvee]{\Rxb{\ct{A \cvee B}}}{
		\Rxb{\ct{A}, \ct{B}}
	}
	\quad
	\infer[\lf\cimp]{\Rxb{\lf{A \cimp B}}}{
		\Rx*{\rt A}
		&
		\Rxb{\lf B}}
	\]
	\[
	\infer[\ct\cimp]{\Rxb{\ct{A \cimp B}}}{
		\Rxb{\lf A, \ct{B}}
	}
	\quad
	\infer[\lf\BOX]{\Dx[1]{\lf{\BOX A}, \BR{\Dx[2]}}}{ 
		\Dx[1]{\lf{\BOX A}, \BR{\lf A, \Dx[2]}}
	}
	\quad
	\infer[\rt\BOX]{\Lx{\rt{\BOX A}}}{
		\Lx{\BR{\rt A}}
	}
	\]
	\[
	\infer[\lf\clozenge]{\Rxb{\lf{\clozenge A}}}{
		\Rxb{\BR{\lf A}}}
	\quad
	\infer[\ct\clozenge]{\Dxb[1]{\ct{\clozenge A},\BR{\Dxb[2]}}}{
		\Dxb[1]{\ct{\clozenge A},\BR{\ct{A},\Dxb[2]}}
	}
	\quad
	\infer[\D]{\Rxb{\ct{P}}}{
		\Rx*{\ct{P},\rt{P}}
	} 
	\quad
	\infer[\store]{\Lx{\rt{N}}}{
		\Lx{\ct{N},\rt{\bot}}
	} 
	\]
	\caption{Classical fragment $\n\EK_c$.}\label{fig:nEKc}\end{figure}

%% file: extensions.tex

Classical modal logic $\K$ is defined as propositional classical logic, extended with the \emph{necessitation rule} (presented in Hilbert style)
$
A/\square A
$
and the \emph{distributivity axiom}
$
\ka:\;\square(A\ra B)\ra (\square A\ra\square B) 
$. 

There are, however, many variants of axiom $\ka$ that induce logics that are classically, but not intuitionistically, equivalent (see~\cite{plotkin:stirling:86,Sim94}).  In fact, the following 
axioms follow from $\ka$ via the De Morgan laws, 
but are intuitionistically independent
\[
\begin{array}{lc@{\qquad}l}
\ka_1:\;\square(A\ra B)\ra (\lozenge A\ra\lozenge B)  & & \ka_2:\;\lozenge(A\vee B)\ra (\lozenge A\vee\lozenge B)\\
\ka_3:\;(\lozenge A\ra \square B)\ra \square(A\ra B) & &  \ka_4:\;\lozenge \bot\ra\bot
\end{array}
\]
Combining axiom $\ka$ with axioms $\ka_1-\ka_4$ defines intuitionistic modal logic $\IK$~\cite{plotkin:stirling:86}. 

In the ecumenical setting, this discussion is even more interesting, since there are many more variants of $\ka$, depending on the classical or intuitionistic interpretation of  implications and diamonds. 

It is an easy exercise to show that the intuitionistic versions of $\ka_1-\ka_4$ are provable in $\nEK$.
One could then ask: what happens if we exchange the intuitionistic versions of the connectives with classical ones? 

Consider $\ka^{\alpha\beta\gamma}:\square(A\to_\alpha B) \to_\beta (\square A \to_\gamma \square B)$ with $\alpha,\beta,\gamma\in\{i,c\}$. First of all, note that $\ka^{c\beta\gamma}$ is not provable, for any $\beta,\gamma$. This is a consequence of the fact that $A\cimp B,A\not\seq B$ in $\EK$ in general (see Equation~\ref{mp}).
Moreover, since $C\iimp D\seq C\cimp D$ in $\EK$, 
$\ka^{\alpha ii}\seq \ka^{\alpha\beta\gamma}$ for any value of $\beta,\gamma$. The same reasoning can be extended to all the other axioms, for example, $\ka_3^{\alpha\beta\gamma\delta}=(\lozenge_\alpha A\to_\beta \square B)\ra_\gamma \square(A\ra_\delta B)$
is not provable for $\beta=c$ and $\ka_3^{iiii}$ implies all the other possible configurations for $\alpha,\beta,\gamma,\delta$. 

Hence, the intuitionistic version of the $\ka$ family of axioms forms their {\em minimal} version valid in $\EK$. In~\cite{DBLP:conf/wollic/MarinPPS21}, we proved that $\nEK$ was cut-complete w.r.t to $\EK$'s Hilbert system based on this set of axioms.

Regarding modal extensions of $\EK$, we can obtain them by restricting the class of frames we consider or, equivalently, by adding axioms over modalities. 
Many of the restrictions one can be interested in are definable as formulas of first-order logic, where the binary predicate $\rel(x,y)$ refers to the corresponding accessibility relation. 
Table~\ref{tab:axioms-FOconditions} summarizes some of the most common logics, the corresponding frame property, together with the modal axiom capturing it~\cite{Sah75}. 

\begin{table}[t]
\begin{center}
  \begin{tabular}{|c|c|c|}
  \hline
    \textbf{Axiom} & \textbf{Condition} & \textbf{First-Order Formula} \\
  \hline
    $\mathsf t:\,\Box A \rightarrow A \wedge A \rightarrow \lozenge A$ & Reflexivity & $\forall x. \rel(x,x)$ \\
    \hline
    $\mathsf b:\,A \rightarrow \Box \lozenge A \wedge\lozenge \Box A \rightarrow A$ & Symmetry & $\forall x,y. \rel(x,y) \rightarrow \rel(y,x)$ \\
\hline
    $\mathsf 4:\,\Box A \rightarrow \Box \Box A \wedge \lozenge\lozenge A \rightarrow \lozenge A$ & Transitivity & $\forall x,y,z. (\rel(x,y) \wedge \rel(y,z)) \rightarrow \rel(x,z)$ \\
\hline
    $\mathsf 5:\,\Box A \rightarrow \Box \lozenge A \wedge \lozenge\Box A \rightarrow \lozenge A$ & Euclideaness & $\forall x,y,z. (\rel(x,y) \wedge \rel(x,z)) \rightarrow \rel(y,z)$ \\
\hline
  \end{tabular}
  \end{center}
  \caption{Axioms and corresponding first-order conditions on $\rel$.}
    \label{tab:axioms-FOconditions}  
\end{table}

\begin{figure}[t]
	\[
\begin{array}{c@{\quad}c@{\quad}c@{\quad}c}
	\infer[\lf{\mathsf{t}}]{\Rx{\lf{\Box A}}}{
		\Rx{\lf{\Box A},\lf{A}}}
	&
	\infer[\lf{\mathsf{b}}]{\Dx[1]{\BR{\Dx[2],\lf{\Box A}}}}{
		\Dx[1]{\BR{\Dx[2],\lf{\Box A}},\lf{A}}}
	&
	\infer[\lf{\mathsf{4}}]{
		\Dx[1]{\BR{\Dx[2]},\lf{\Box A}}}{\Dx[1]{\BR{\Dx[2],\lf{\Box A}},\lf{\Box A}}}
	&
	\infer[\lf{\mathsf{5}}]{
		\Rx{\BR{\lf{\Box A}}\BR{\emptyset}}}{
		\Rx{\BR{\lf{\Box A}}\BR{\lf{\Box A}}}}
\\[1em]
	\infer[\rt{\mathsf{t}}]{\Lx{\rt{\ilozenge A}}}{\Lx{\rt{A}}}
	&
	\infer[\rt{\mathsf{b}}]{\Lx[1]{\BR{\Lx[2],\rt{\ilozenge A}}}}{
		\Lx[1]{\BR{\Lx[2]},\rt{A}}}
	&
	\infer[\rt{\mathsf{4}}]{
		\Lx[1]{\BR{\Lx[2]},\rt{\ilozenge A}}}{\Lx[1]{\BR{\Lx[2],\rt{\ilozenge A}}}}
	&
	\infer[\rt{\mathsf{5}}]{
		\Lx{\BR{\rt{\ilozenge A}}\BR{\emptyset}}}{
		\Lx{\BR{\emptyset}\BR{\rt{\ilozenge A}}}}
\\[1em]
	\infer[\ct{\mathsf{t}}]{\Rxb{\ct{\clozenge A}}}{\Rxb{\ct{A}}}
	&
	\infer[\ct{\mathsf{b}}]{\Dxb[1]{\BR{\Dxb[2],\ct{\clozenge A}}}}{
		\Dxb[1]{\BR{\Dxb[2]},\ct{A}}}
	&
	\infer[\ct{\mathsf{4}}]{
		\Dxb[1]{\BR{\Dxb[2]},\ct{\clozenge A}}}{\Dxb[1]{\BR{\Dxb[2],\rt{\ilozenge A}}}}
	&
	\infer[\ct{\mathsf{5}}]{
		\Rxb{\BR{\ct{\clozenge A}}\BR{\emptyset}}}{
		\Rxb{\BR{\emptyset}\BR{\ct{\clozenge A}}}}
\end{array}
	\]
	\caption{Ecumenical modal extensions for axioms $\te,\be,\4$ and $\5$.}\label{fig:ext}
\end{figure}

Since the intuitionistic fragment of $\nEK$ coincides with $\NIK$, intuitionistic versions for the rules for the axioms $\te, \be, \4$, and $\5$ match the rules ($\lf{}$) and ($\rt{}$) presented in~\cite{DBLP:conf/fossacs/Strassburger13}, and are depicted in Figure~\ref{fig:ext}.

For completing the ecumenical view, the classical ($\ct{}$) rules for extensions are justified via translations from labeled systems to $\nEK$: We first translate the labeled rules for extensions appearing in~\cite{Sim94} to $\labEK$ then use the translation on derivations defined in Section~\ref{sec:sound-comp} to justify the rule scheme. 

For example, starting with the rule $\T$ below left, which is the labeled rule corresponding to the axiom $\te$ in~\cite{Sim94}, the labeled derivation on the middle  justifies the classical nested rule in the right.
\[
\infer[\T]{\Gamma\vdash z:C}{xRx,\Gamma\vdash z:C}
 \qquad
\infer[\T]{\mathcal{R},\Sigma\seq \Delta ,x: \clozenge A; z:\bot}
{\infer[\clozenge R]{xRx,\mathcal{R},\Sigma\seq \Delta ,x: \clozenge A; z:\bot}
{\deduce{xRx,\mathcal{R},\Sigma\seq \Delta ,x:A,x: \clozenge A; z:\bot}{}}} \qquad 
\infer[\ct{\mathsf{t}}]{\Rxb{\ct{\clozenge A}}}{\Rxb{\ct{A}}}
\]
The rules $\ct\be,\ct\4$ and $\ct\5$, shown in Figure~\ref{fig:ext}, are obtained in the same manner. 

Restricted to the fragments described in the last section, by mixing and matching these rules, we obtain ecumenical modal systems for the logics in the $\Sfi$ modal cube~\cite{blackburn_rijke_venema_2001}  not defined with axiom $\mathsf{d}$.


%% file: conclusion.tex

The main idea behind Prawitz' ecumenical system~\cite{DBLP:journals/Prawitz15} is to build a proof framework in which classical and intuitionistic logics may co-exist in peace. Although one could argue that this is easily done using the well known double-negation translations by Kolmogorov, G\"{o}del, Gentzen and others~\cite{DBLP:journals/corr/abs-1101-5442}, Prawitz' view matches the idea  presented by Liang and Miller in their $PIL$ system presented in~\cite{DBLP:journals/apal/LiangM11}:  not seeing classical logic as a fragment of intuitionistic logic but rather to determine parts of {\em reasoning} which are classical or intuitionistic in nature. While double negation acts on {\em formulas}, the approach in~\cite{DBLP:journals/apal/LiangM11} and also followed here concerns {\em proofs}. For example, we {\em do not} want to interpret $A\vee\neg A$ as 
\begin{quote}``it is not the case that $A$ does not hold and it is not the case that it is not the case that $A$ holds''.
\end{quote}
Rather, we aim at identifying the points in proofs where the excluded middle is valid and/or necessary.

The similarities between our work and the system presented in~\cite{DBLP:journals/apal/LiangM11} ends there, though.  Indeed, in the {\em op.cit.} there are two versions of  the constant for absurdum and universal quantifier, and all connectives have a dual version. For example, the intuitionistic implication $\supset$ comes with the intuitionistic dual $\varpropto$, a form of (non-commutative) conjunction, which has no correspondent in usual classical or intuitionistic systems. Also, these dualized versions have opposite {\em polarities} (red and green), that do not match Girard's original idea of polarities:  They are, instead, defined model theoretically. In this work, we opt for smoothly extending well known systems and features (like stoup or polarities), which turns $\LCE$ and $PIL$  incomparable. 
It would be interesting to investigate, for example, if $PIL$ could be smoothly extended to the modal case, as done in this work.

There are other proposals for ecumenical systems in the literature. For, in~\cite{DBLP:conf/fscd/BlanquiDGHT21} the authors present a (type) theory in $\lambda\Pi$-calculus modulo theory, where proofs of several logical systems can be expressed. We are planning to propose type systems related to the systems/fragments described in this paper, and it would be interesting to see the intersection that may appear from the two approaches. It would be also interesting  to implement ecumenical provers, as well as to automate the cut-elimination proof in the L-Framework~\cite{DBLP:journals/corr/abs-2101-03113}.

A complete different approach comes from the school of combining logics~\cite{DBLP:conf/frocos/CerroH96,DBLP:conf/frocos/Lucio00,DBLP:conf/frocos/CaleiroR07}, where Hilbert like systems are built from a combination of axiomatic systems. 
As we trail the exact opposite path, it would be interesting to see if (the propositional fragment of) Prawitz' natural deduction system is axiomatizable.  

Finally, the presence of polarization and stoup paves the way for proposing focused ecumenical systems. For getting a complete focused discipline, though, it would be necessary to add polarized versions of conjunction and disjunction, as done \eg\ in~\cite{DBLP:journals/apal/LiangM11,DBLP:conf/rta/ChaudhuriMS16}. This would give a unified focused framework, which could be used, among other things, to automatically extracting rules from axioms, as done in~\cite{MARIN2022103091}.

%% file: main-JLC.bbl
\begin{thebibliography}{10}

\bibitem{andreoli01apal}
J.-M. Andreoli.
\newblock Focussing and proof construction.
\newblock {\em Annals of Pure and Applied Logic}, 107(1):131--163, 2001.

\bibitem{blackburn_rijke_venema_2001}
P.~Blackburn, M.~d. Rijke, and Y.~Venema.
\newblock {\em Modal Logic}.
\newblock Cambridge Tracts in Theoretical Computer Science. Cambridge
  University Press, 2001.

\bibitem{DBLP:conf/fscd/BlanquiDGHT21}
F.~Blanqui, G.~Dowek, {\'{E}}.~Grienenberger, G.~Hondet, and F.~Thir{\'{e}}.
\newblock Some axioms for mathematics.
\newblock In N.~Kobayashi, editor, {\em 6th International Conference on Formal
  Structures for Computation and Deduction, {FSCD} 2021, July 17-24, 2021,
  Buenos Aires, Argentina (Virtual Conference)}, volume 195 of {\em LIPIcs},
  pages 20:1--20:19. Schloss Dagstuhl - Leibniz-Zentrum f{\"{u}}r Informatik,
  2021.

\bibitem{Brunnler:2009kx}
K.~Br{\"u}nnler.
\newblock Deep sequent systems for modal logic.
\newblock {\em Arch. Math. Log.}, 48:551--577, 2009.

\bibitem{DBLP:conf/frocos/CaleiroR07}
C.~Caleiro and J.~Ramos.
\newblock Combining classical and intuitionistic implications.
\newblock In B.~Konev and F.~Wolter, editors, {\em Frontiers of Combining
  Systems, 6th International Symposium, FroCoS 2007, Liverpool, UK, September
  10-12, 2007, Proceedings}, volume 4720 of {\em Lecture Notes in Computer
  Science}, pages 118--132. Springer, 2007.

\bibitem{DBLP:conf/rta/ChaudhuriMS16}
K.~Chaudhuri, S.~Marin, and L.~Stra{\ss}burger.
\newblock Modular focused proof systems for intuitionistic modal logics.
\newblock In {\em 1st International Conference on Formal Structures for
  Computation and Deduction, {FSCD} 2016, June 22-26, 2016, Porto, Portugal},
  pages 16:1--16:18, 2016.

\bibitem{DBLP:conf/frocos/CerroH96}
L.~F. del Cerro and A.~Herzig.
\newblock Combinig classical and intuitionistic logic, or: Intuitionistic
  implication as a conditional.
\newblock In F.~Baader and K.~U. Schulz, editors, {\em Frontiers of Combining
  Systems, First International Workshop FroCoS 1996, Munich, Germany, March
  26-29, 1996, Proceedings}, volume~3 of {\em Applied Logic Series}, pages
  93--102. Kluwer Academic Publishers, 1996.

\bibitem{DBLP:conf/ictac/Diaz-CaroD21}
A.~D{\'{\i}}az{-}Caro and G.~Dowek.
\newblock A new connective in natural deduction, and its application to quantum
  computing.
\newblock In A.~Cerone and P.~C. {\"{O}}lveczky, editors, {\em Theoretical
  Aspects of Computing - {ICTAC} 2021 - 18th International Colloquium, Virtual
  Event, Nur-Sultan, Kazakhstan, September 8-10, 2021, Proceedings}, volume
  12819 of {\em Lecture Notes in Computer Science}, pages 175--193. Springer,
  2021.

\bibitem{DBLP:journals/Dowek16a}
G.~Dowek.
\newblock On the definition of the classical connectives and quantifiers.
\newblock {\em Why is this a Proof?, Festschrift for Luiz Carlos Pereira},
  27:228--238, 2016.

\bibitem{DBLP:journals/logcom/DyckhoffL07}
R.~Dyckhoff and S.~Lengrand.
\newblock Call-by-value lambda-calculus and {LJQ}.
\newblock {\em J. Log. Comput.}, 17(6):1109--1134, 2007.

\bibitem{DBLP:journals/corr/abs-1101-5442}
G.~Ferreira and P.~Oliva.
\newblock On various negative translations.
\newblock In S.~van Bakel, S.~Berardi, and U.~Berger, editors, {\em Proceedings
  Third International Workshop on Classical Logic and Computation, CL{\&}C
  2010, Brno, Czech Republic, 21-22 August 2010}, volume~47 of {\em {EPTCS}},
  pages 21--33, 2010.

\bibitem{Fitting:2014}
M.~Fitting.
\newblock Nested sequents for intuitionistic logics.
\newblock {\em Notre Dame Journal of Formal Logic}, 55(1):41--61, 2014.

\bibitem{DBLP:journals/mscs/Girard91}
J.~Girard.
\newblock A new constructive logic: Classical logic.
\newblock {\em Math. Struct. Comput. Sci.}, 1(3):255--296, 1991.

\bibitem{DBLP:journals/apal/Girard93}
J.~Girard.
\newblock On the unity of logic.
\newblock {\em Ann. Pure Appl. Logic}, 59(3):201--217, 1993.

\bibitem{DBLP:conf/aiml/GoreR12}
R.~Gor{\'{e}} and R.~Ramanayake.
\newblock Labelled tree sequents, tree hypersequents and nested (deep)
  sequents.
\newblock In T.~Bolander, T.~Bra{\"{u}}ner, S.~Ghilardi, and L.~S. Moss,
  editors, {\em Advances in Modal Logic 9, papers from the ninth conference on
  "Advances in Modal Logic," held in Copenhagen, Denmark, 22-25 August 2012},
  pages 279--299. College Publications, 2012.

\bibitem{DBLP:conf/csl/Herbelin94}
H.~Herbelin.
\newblock A lambda-calculus structure isomorphic to gentzen-style sequent
  calculus structure.
\newblock In L.~Pacholski and J.~Tiuryn, editors, {\em Computer Science Logic,
  8th International Workshop, {CSL} '94, Kazimierz, Poland, September 25-30,
  1994, Selected Papers}, volume 933 of {\em Lecture Notes in Computer
  Science}, pages 61--75. Springer, 1994.

\bibitem{DBLP:journals/synthese/KahleS06}
R.~Kahle and P.~Schroeder{-}Heister.
\newblock Introduction: Proof-theoretic semantics.
\newblock {\em Synth.}, 148(3):503--506, 2006.

\bibitem{DBLP:conf/tableaux/Lellmann19}
B.~Lellmann.
\newblock Combining monotone and normal modal logic in nested sequents - with
  countermodels.
\newblock In {\em TABLEAUX}, volume 11714 of {\em LNCS}, pages 203--220, 2019.

\bibitem{DBLP:journals/apal/LiangM11}
C.~Liang and D.~Miller.
\newblock A focused approach to combining logics.
\newblock {\em Ann. Pure Appl. Logic}, 162(9):679--697, 2011.

\bibitem{DBLP:conf/frocos/Lucio00}
P.~Lucio.
\newblock Structured sequent calculi for combining intuitionistic and classical
  first-order logic.
\newblock In H.~Kirchner and C.~Ringeissen, editors, {\em Frontiers of
  Combining Systems, Third International Workshop, FroCoS 2000, Nancy, France,
  March 22-24, 2000, Proceedings}, volume 1794 of {\em Lecture Notes in
  Computer Science}, pages 88--104. Springer, 2000.

\bibitem{MARIN2022103091}
S.~Marin, D.~Miller, E.~Pimentel, and M.~Volpe.
\newblock From axioms to synthetic inference rules via focusing.
\newblock {\em Annals of Pure and Applied Logic}, 173(5):103091, 2022.

\bibitem{DBLP:conf/dali/MarinPPS20}
S.~Marin, L.~C. Pereira, E.~Pimentel, and E.~Sales.
\newblock Ecumenical modal logic.
\newblock In M.~A. Martins and I.~Sedl{\'{a}}r, editors, {\em Dynamic Logic.
  New Trends and Applications - Third International Workshop, DaL{\'{\i}} 2020,
  Prague, Czech Republic, October 9-10, 2020, Revised Selected Papers}, volume
  12569 of {\em Lecture Notes in Computer Science}, pages 187--204. Springer,
  2020.

\bibitem{DBLP:conf/wollic/MarinPPS21}
S.~Marin, L.~C. Pereira, E.~Pimentel, and E.~Sales.
\newblock A pure view of ecumenical modalities.
\newblock In A.~Silva, R.~Wassermann, and R.~J. G.~B. de~Queiroz, editors, {\em
  Logic, Language, Information, and Computation - 27th International Workshop,
  WoLLIC 2021, Virtual Event, October 5-8, 2021, Proceedings}, volume 13038 of
  {\em Lecture Notes in Computer Science}, pages 388--407. Springer, 2021.

\bibitem{DBLP:journals/tcs/MillerP13}
D.~Miller and E.~Pimentel.
\newblock A formal framework for specifying sequent calculus proof systems.
\newblock {\em Theor. Comput. Sci.}, 474:98--116, 2013.

\bibitem{Murzi2018}
J.~Murzi.
\newblock Classical harmony and separability.
\newblock {\em Erkenntnis}, 2018.

\bibitem{DBLP:journals/corr/abs-2101-03113}
C.~Olarte, E.~Pimentel, and C.~Rocha.
\newblock A rewriting logic approach to specification, proof-search, and
  meta-proofs in sequent systems.
\newblock {\em CoRR}, abs/2101.03113, 2021.

\bibitem{luiz17}
L.~C. Pereira and R.~O. Rodriguez.
\newblock Normalization, soundness and completeness for the propositional
  fragment of {P}rawitz' ecumenical system.
\newblock {\em Revista Portuguesa de Filosofia}, 73(3-3):1153--1168, 2017.

\bibitem{DBLP:journals/synthese/PimentelPP21}
E.~Pimentel, L.~C. Pereira, and V.~de~Paiva.
\newblock An ecumenical notion of entailment.
\newblock {\em Synthese}, 198(22-S):5391--5413, 2021.

\bibitem{plotkin:stirling:86}
G.~D. Plotkin and C.~P. Stirling.
\newblock A framework for intuitionistic modal logic.
\newblock In J.~Y. Halpern, editor, {\em 1st Conference on Theoretical Aspects
  of Reasoning About Knowledge}. Morgan Kaufmann, 1986.

\bibitem{Poggiolesi:2009vn}
F.~Poggiolesi.
\newblock The method of tree-hypersequents for modal propositional logic.
\newblock In {\em Towards Mathematical Philosophy}, volume~28 of {\em Trends In
  Logic}, pages 31--51. Springer, 2009.

\bibitem{DBLP:journals/Prawitz15}
D.~Prawitz.
\newblock Classical versus intuitionistic logic.
\newblock {\em Why is this a Proof?, Festschrift for Luiz Carlos Pereira},
  27:15--32, 2015.

\bibitem{Sah75}
H.~Sahlqvist.
\newblock Completeness and correspondence in first and second order semantics
  for modal logic.
\newblock In N.~H. S.~Kanger, editor, {\em Proceedings of the Third
  Scandinavian Logic Symposium}, pages 110--143, 1975.

\bibitem{Sim94}
A.~K. Simpson.
\newblock {\em The Proof Theory and Semantics of Intuitionistic Modal Logic}.
\newblock PhD thesis, College of Science and Engineering, School of
  Informatics, University of Edinburgh, 1994.

\bibitem{DBLP:conf/fossacs/Strassburger13}
L.~Stra{\ss}burger.
\newblock Cut elimination in nested sequents for intuitionistic modal logics.
\newblock In {\em Proceedings of FOSSACS 2013}, pages 209--224, 2013.

\bibitem{troelstra96bpt}
A.~S. Troelstra and H.~Schwichtenberg.
\newblock {\em Basic Proof Theory}.
\newblock Cambridge Univ. Press, 1996.

\end{thebibliography}
